\pdfoutput=1
\documentclass[final,onefignum,onetabnum,bookmarks=false]{siamart171218}

\usepackage{amsmath,amssymb,cite,graphicx,color}
\usepackage{fullpage}
\usepackage[subrefformat=parens,labelformat=parens,font=normalsize]{subfig}
\usepackage{float}
\usepackage[fleqn]{nccmath}
\usepackage{ dsfont }
\usepackage{multirow}
\usepackage{url}

\newcommand{\R}{\mathbb{R}}
\newcommand{\C}{\mathbb{C}}

\renewcommand{\d}[1]{\mathrm{d}#1}

\newcommand{\e}{\mathrm{e}}
\renewcommand{\j}{\mathrm{j}}

\newcommand{\calR}{\mathcal{R}}
\newcommand{\calS}{\mathcal{S}}

\DeclareMathOperator{\argmin}{arg \ min}
\DeclareMathOperator{\row}{row}
\newcommand\blfootnote[1]{%
  \begingroup
  \renewcommand\thefootnote{}\footnote{#1}%
  \addtocounter{footnote}{-1}%
  \endgroup
}
\graphicspath{{./Figs/}}
\begin{document}
\title{Trading beams for bandwidth: Imaging with randomized beamforming}
\author{Rakshith Sharma Srinivasa, Mark A.\ Davenport, Justin Romberg \\ Georgia institute of technology}

\maketitle

\begin{abstract}
We study the problem of actively imaging a range-limited far-field scene using an antenna array.  We describe how the range limit imposes structure in the measurements across multiple wavelengths.  This structure allows us to introduce a novel trade-off: the number of spatial array measurements (i.e., beams that have to be formed) can be reduced significantly lower than the number array elements if the scene is illuminated with a broadband source.  To take advantage of this trade-off, we use a small number of ``generic" linear combinations of the array outputs, instead of the phase offsets used in conventional beamforming. 
We provide theoretical justification for the proposed trade-off without making any strong structural assumptions on the target scene (such as sparsity) except that it is range limited. In proving our theoretical results, we take inspiration from the \textit{sketching} literature. We also provide simulation results to establish the merit of the proposed signal acquisition strategy. Our proposed method results in a reduction in the number of required spatial measurements in an array imaging system and hence can directly impact their speed and cost of operation. 

\end{abstract}
\section{Introduction}
\label{sec:Intro}
In active array imaging, a transmitter emits an excitation signal and then forms an image using the reflections collected at an array of sensors. This technique has been employed in a multitude of fields ranging from medicine to security and surveillance, among many others. Array imaging offers a window outside of the visible spectrum, which can prove crucial in many applications where visible light cannot penetrate. However, a high cost barrier has historically prevented widespread adoption of array imaging in commercial products \cite{ehyaie2011novel}. Nevertheless, applications such as autonomous vehicles, depth sensing, gesture recognition \cite{soli} and others have caused an increase of interest in commercializing active imaging modalities such as LiDAR and RADAR, leading to a number of efforts aimed at reducing the cost and increasing the efficiency of array imaging systems. In this paper, we consider a general antenna array imaging system and propose a novel trade-off that utilizes bandwidth of the excitation signal to reduce the number of array measurements needed to image targets that are range-limited, thus directly impacting the speed and cost of acquisition.

\blfootnote{This work was supported, in part, by NSF grants CCF-1350616 and CCF-1409406, a grant from Lockheed Martin, and a gift from the Alfred P.\ Sloan Foundation}

\blfootnote{A part of the work presented here appeared in \cite{SDR_Sketching_2017} and in \cite{SDR_Localized_2018}}

We use the standard linear model for array imaging in a setting with a single transmitter and multiple receiver elements. We use $M$ to denote the number of antenna array elements and $\lambda$ to denote the wavelength of the active excitation signal. We will later show that the set of discrete measurements collected by the antenna array elements at wavelength $\lambda$, $y_\lambda$, can be modeled as 
\begin{equation}
\label{eq:intro_linear_model}
y_\lambda = A_\lambda x_0
\end{equation} 
where $x_0 \in \R^N$ denotes the sampled target scene and $A_\lambda$ is the linear operator ($M \times N$ matrix) mapping the scene to the measurements.  Collecting broadband measurements is equivalent to collecting $y = \{ y_\lambda\}$ for ${\lambda = \lambda_1, \ \lambda_2, \ \cdots, \ \lambda_K}$. We note that all $\{y_\lambda \}$ can be obtained simultaneously (using only $M$ ``spatial measurements'') by using a single broadband excitation signal and computing a temporal Fourier transform. 

A standard imaging method known as \textit{beamforming} collects linear combinations of the array elements' outputs instead of sampling them directly.  This can be modeled as:
\begin{equation}
\label{eq:intro_beamforming}
z_\lambda = \psi_\lambda y_\lambda = \psi_\lambda A_\lambda x_0 
\end{equation} where $\psi_\lambda$ is typically an invertible $M \times M$ matrix. In traditional beamforming the weights are chosen to induce spatial selectivity where each of the $M$ measurements collect reflections from distinct spatial sectors. This can be considered a special case of \textit{aperture coding}: collecting linear combinations of the array outputs.  Henceforth, we refer to linear combinations of array outputs as spatial measurements or beams. Note that beams conventionally refer to directivity inducing linear combinations, but for the purposes of this paper, we refer to any linear combination of the array outputs as a beam.

Our main observation is that the set of spatial, broadband measurements of such range-limited targets lie in a lower dimensional subspace and hence have a limited number of degrees of freedom. This naturally leads to the question of whether spatially undersampled array data of such scenes can still yield reconstructions as good as those obtained with full data. We provide a positive answer to this question. While this is reminiscent of {\em compressed sensing}, an important distinction between our work and the compressed sensing paradigm is that we do not impose any model such as sparsity on the target scene. We only require the scene to be range-limited.  

Our motivation in this work is based on the observation that sampling and storing all of the array elements' outputs, or obtaining all of the $M$ possible linear combinations required for traditional beamforming can be challenging and is in fact wasteful when the target scene is range-limited. We show that in this case, by taking fewer \textit{generic} linear combinations (or \textit{aperture codes}) or even just spatial subsampling, one can obtain the same quality reconstructions as that of using full measurements.

We take inspiration from a set of dimensionality reduction techniques used in the numerical linear algebra community known as \textit{sketching} \cite{woodruff2014sketching} to provide theoretical justification for the proposed signal acquisition method. We show that image reconstruction with a few generic aperture codes is equivalent to a \textit{sketched least squares} problem of the form \[ \min_x \|\Phi y - \Phi Ax \|^2 \] where $\Phi$ is a highly structured compresssive matrix. We show that it provides a solution equivalent to that of solving the higher dimensional problem \[ \min_x \|y - Ax \|^2 \] which corresponds to standard image reconstruction methods. However, we again emphasize that unlike compressed sensing techniques, we do not assume any sparsity in the image domain.

\begin{figure}[!t]
\centering
\begin{minipage}{.35\linewidth}
\centering
\subfloat[\small \sl full imaging]{\label{CR:a}\includegraphics[scale=.2]{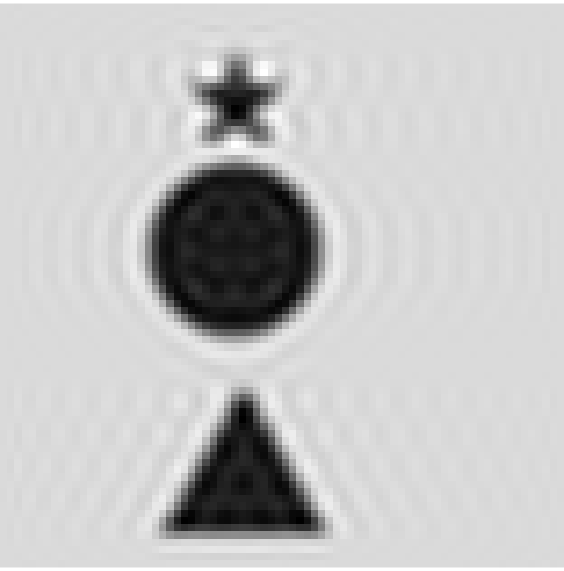}}
\end{minipage}%
\hspace{2em}
\begin{minipage}{.35\linewidth}
\centering
\subfloat[ \small \sl 320 generic beams]{\label{CR:c}\includegraphics[scale=.2]{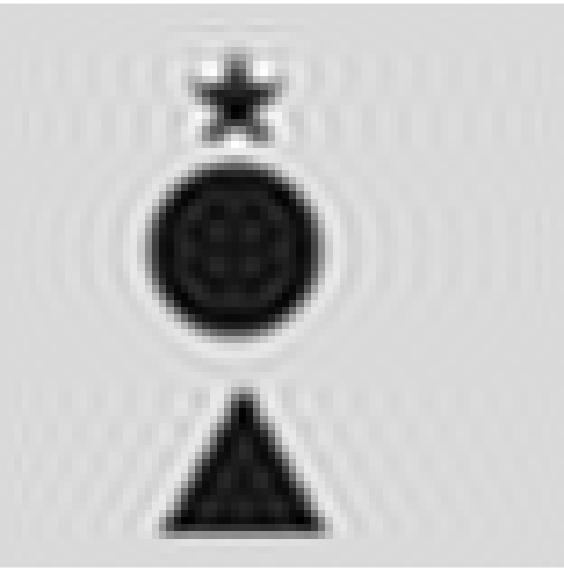}}
\end{minipage}
\par\medskip
\begin{minipage}{.35\linewidth}
\centering
\subfloat[\small \sl 160 generic beams]{\label{CR:d}\includegraphics[scale=.2]{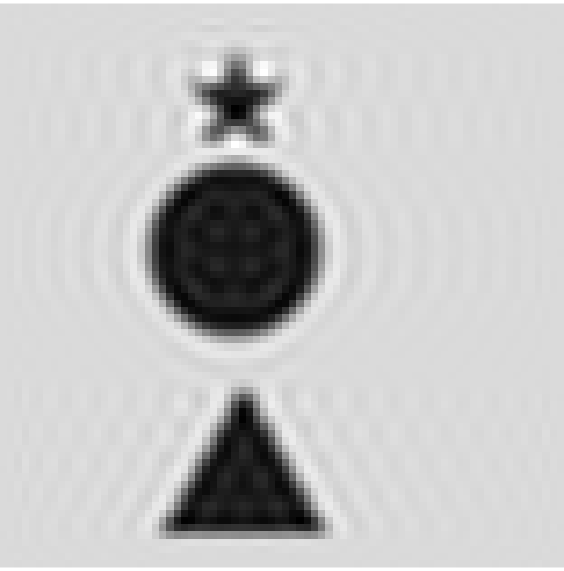}}
\end{minipage}
\hspace{2em}% 
\begin{minipage}{.35\linewidth}
\centering
\subfloat[\small \sl 80 generic beams]{\label{CR:e}\includegraphics[scale=.2]{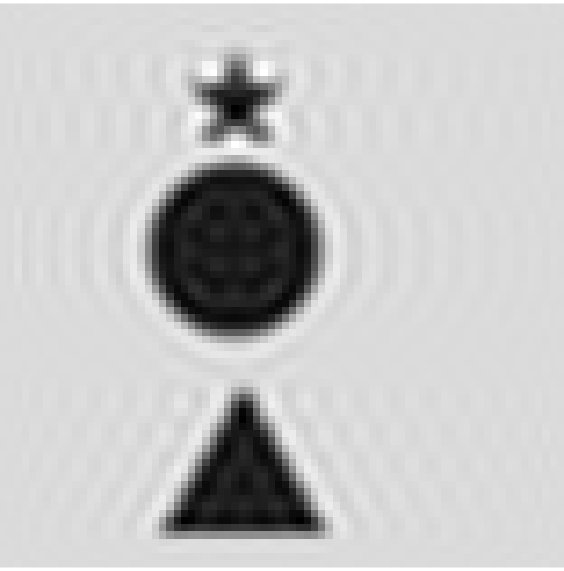}}
\end{minipage}

\caption{\small \sl Aperture coded imaging for images at a constant range. (a) represents the conventional method, which uses 1100 beams. (b),(c) and (d)  show the reconstruction results using only $320$, $160$ and $80$ generic linear combinations of the antenna array outputs.  }
\label{fig:CRmain}
\end{figure}

As an example, consider an imaging setup where the array is two dimensional and has $40 \times 40$ sensors with sensors placed $\lambda/2 = 3.75$cm apart and a target scene having a span of $[-45^{\circ}, 45^{\circ}]$ in both elevation and azimuthal angles. Let the scene have \textit{delta thickness}: only one reflector per each angle, present at a constant known depth. Standard ways of imaging such a scene would require around $1100$ beams at wavelength $\lambda = 7.5$cm. By introducing bandwidth in the excitation signal, we show that the scene can be imaged with as few as $80$ spatial linear combinations. This is illustrated in Figure \ref{fig:CRmain}. For target scenes with higher range limits, standard imaging methods also need bandwidth for imaging \cite{6600891}. We however show that this bandwidth can be used to obtain similar gains in the number of spatial linear combinations and provide theoretical justification for the gains that aperture coding/subsampling can provide.

\section{Related work}
\label{sec:rel_work}

% Standard references on array imaging

Array imaging has been addressed in a number of previous works.  Standard array imaging problems have been considered in \cite{6600891, 148610} where the inverse problem is set up using transmit and receive beamforming with narrowband excitation. The authors also describe the fundamental limits on the resolution of array imaging systems arising due to finite apertures and discrete sampling of the array. General 3D imaging with wide-band excitation is considered in \cite{942570}. In general, to identify a 3D scene, a 2D antenna array and wideband excitation are necessary. Our focus is on a different dimension of this problem: we identify the limited number of degrees of freedom in a range-limited image and explore an alternate way using fewer measurements in  which such a scene can be captured. 

% References to papers that try to reduce the cost of array imaging systems.
\par As arrays get larger, or when arrays need to be used in low-cost commercial applications, it is desirable to reduce the reduce the number of elements in an array and to use smarter algorithms to reduce the cost of the system. By sequentially using different parts of the array, one can obtain a set of low resolution images and propose to achieve reconstruction by upsampling and summing these images~\cite{1397349, 1397350}. In another approach, using a carefully designed non-uniform array, \cite{5456168} proposes to increase the number of resolvable directions to $O(N^2)$ by using an array with $O(N)$ elements. This approach enables the use of specially designed arrays to solve the problem of direction of arrival of source signals in passive sensing scenarios. Another main theme in reducing the cost of array systems has been the use of compressed sensing techniques. Reducing sampling rates at the sensors for digital beamforming is proposed in \cite{6854941, 6863846, 6203608}, but the number of beams/ array elements remain the same as conventional imaging. In a slightly different application, \cite{Wiaux21052009} imposes sparsity based regularization to solve the ill-posed radio-interferometric imaging problem. A similar theme is also followed in \cite{7833048} which addresses the same problem of 3D imaging with a 2D phased array, but assumes sparsity in the image domain. The trade-off we propose has a different flavor: we demonstrate that range-limited images have a limited number of degrees of freedom and impose no further models such as sparsity or total variation on the image. Unlike in \cite{6203608, 6863846, 6854941}, we directly address the number of measurements/array elements rather than sampling rates at each element.

% Reference to papers on spatio-spectral concentration 
The problem of understanding the degrees of freedom in the context of active imaging of range-limited scenes is a particular case of the phenomenon of simultaneous concentration of energy of a signal in spatial/temporal and spectral domains. Signals which are concentrated maximally in a given time interval and a frequency interval are well-approximated using a subspace of dimension approximately equal to the product of the lengths of the intervals. This has been well studied for one dimensional signals  \cite{1454379, BLTJ:BLTJ3976, BLTJ:BLTJ3977,BLTJ:BLTJ3279,KZWRD_FST_2019,ZKWDR_Roast_2018,ZKDRW_eigenvalue_2018}. Spatio-spectral concentration in two dimensions has also been studied in \cite{Simons2010,Simons2011}.
We study this phenomenon in the case antenna array imaging and propose methods to drive the number of measurements close to the number of degrees of freedom, in contrast to conventional imaging methods that collect far more number of measurements.

% References to subspace capturing using randomized projections
We model our approach as a novel matrix sketching problem and provide guarantees on the number of measurements/beams required to image a far-field target. Matrix sketching refers to a set of techniques in numerical linear algebra for dimensionality reduction. In particular, a given large matrix is pre- or post-multiplied by a suitable randomized matrix to reduce the ambient dimensions while still retaining the required information, such as the Euclidean distance between elements in its column space. Research in this area has seen a rise in popularity due to its utility in solving large problems in numerical linear algebra~\cite{7355313, woodruff2014sketching }. Our theoretical results are most closely aligned with those in~\cite{Nhalko}, where random Gaussian projections are used to capture the range space of linear operators. In the context of active imaging, usage of sketching ideas can be seen in \cite{doi:10.1093/mnras/stx531} where the Fourier basis is used as a sketching matrix for dimensionality reduction in interferometry. But the sketching matrix in \cite{doi:10.1093/mnras/stx531} is a generic rectangular sketching matrix. In the context of wideband array imaging, the sketching matrix involved itself has a very particular structure that is dictated by the physical setup, which introduces new challenges. The work presented in the remainder of this paper is an extension of ideas first presented in a preliminary form in~\cite{SDR_Localized_2018,SDR_Sketching_2017}.
\section{Active array imaging}
\label{sec:std_array_imaging}
In this section, we describe the standard setup and measurement models of active array imaging. While doing so, we highlight some facts about broadband array imaging that form the  basis of our contribution. In particular, we will emphasize that broadband array measurements provide a set of bandlimited Fourier domain samples of the target scene. We also describe some standard acquisition methods used to physically collect these samples. In the subsequent sections, we show that these bandlimited samples are highly structured for range-limited target scenes. 

\subsection{Propagation model and Fourier domain samples}

\label{subsec:broadband_imaging}

\begin{figure}[!tbhp]
\centering
\begin{minipage}[b]{0.48\textwidth}
\centering
\subfloat[]{\label{fig:CS1d}\includegraphics[scale=0.45]{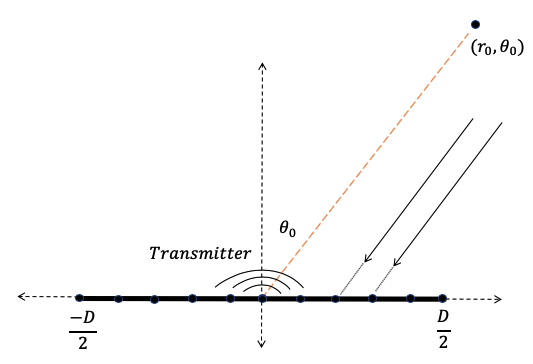}}
\end{minipage}
\hfill
\begin{minipage}[b]{0.48\textwidth}
\centering
\subfloat[]{\label{fig:CS}\includegraphics[scale = 0.3]{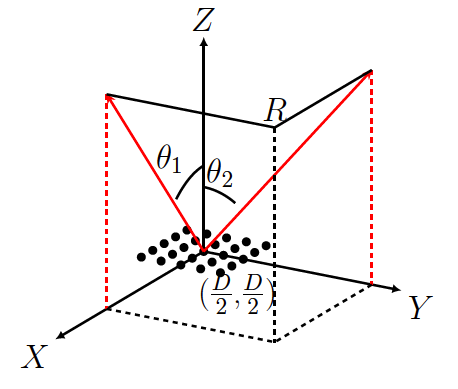}}
\end{minipage}
\caption{\small \sl  (a) shows the array imaging setup in the 1D case. (b) shows  the coordinate system used for the 2D array setup. $R$ represents a point on the target image. The antenna array, illustrated using the set of dots, lies in the X-Y plane in the region $[\frac{-D}{2}, \frac{D}{2}] \times [\frac{-D}{2}, \frac{D}{2}]$.}
\label{fig:setup}
\end{figure}

% For a SIMO 1D array, the array outputs are FT
We first consider the broadband array imaging problem with a one-dimensional (1D) antenna array. This can be easily extended to two dimensions. Consider a uniform linear array of aperture length $D$ placed on the $x$-axis, spanning $[-D/2, D/2]$. For now, we assume that the aperture is continuous. The phase center is taken at the origin: i.e., the time delay of arrival at each element is measured with respect to the element at the origin. We assume that the scene to be imaged lies in the X-Y plane and is in the far-field region of the array. Imaging the scene is equivalent to reconstructing its reflectivity map, which is a function of the roundtrip distance of the target from the array center and the angle from the broadside, to be denoted as $p(r,\theta)$. The system consists of only one transmitting element which is co-located with the receiver at the array center. Let $t$ be the continuous time index. For an excitation signal $s(t)$, the signal received at the array location $d$ is $s(t - r_0/c - d\sin \theta_0 /c)$ for a unit-strength reflector at $(r_0, \theta_0)$. For a general reflectivity map, the narrowband response at this location for the excitation signal $s(t) = e^{j2\pi ct/\lambda}$ is 
\begin{align}
\widetilde{y}_{d,\lambda}(t) = \e^{\j 2\pi ct/\lambda}\! \int\limits_{-\pi/2}^{\pi/2}\int p(r,\theta)\e^{-\j 2\pi (r +{d\sin\theta })/\lambda}\, d r\,d\theta.
\label{eq:narrowband_eqn}
\end{align}

By making the substitution $\tau= (\sin\theta)/2$, the complex amplitude of the signal received at location $d \in [-D/2, D/2]$ for excitation wavelength $\lambda$ can be written as 
\begin{equation}
	y_{d,\lambda} =  \int\limits_{-\pi/2}^{\pi/2}\int p(r,\theta)\e^{-\j 2\pi (r +{d\sin\theta })/\lambda}\, d r\,d\theta = \widehat{x}_c\left(\,\omega_r , \, \omega_\tau \right).
\end{equation}
where $\omega_r = \frac{1}{\lambda}, \ \omega_\tau = \frac{2d}{\lambda}$, and $\widehat{x}_c$ denotes the Fourier transform of $x_c(r,\tau) = \frac{p(r,\sin^{-1}(2\tau))}{\sqrt{1 - 4\tau^2 }}$. This shows that the antenna aperture measures the Fourier transform of (a warped version of) the target scene. $y_{d,\lambda}$ can be obtained by computing the Fourier transform of the temporal signal $\widetilde{y}_{d,\lambda}(t)$ received at the array after sampling it at a suitable rate, or by measuring the complex amplitude of the received signal. As a direct consequence of the finiteness of the aperture,we have that at any excitation wavelength $\lambda$, the accessible interval in the Fourier domain for $\omega_{\tau}$ is limited to $[-D/\lambda, D/\lambda]$.

\par  For an imaging system with a 2D array in the X-Y plane and a 3D scene, the extension of the setup is straightforward to derive. The coordinates in 3D can be denoted as $(r, \theta_1, \theta_2)$, where $r$ is the roundtrip distance to the array center, $\theta_1$ is the angle with respect to the Y-Z plane, and $\theta_2$ is the angle with respect to the X-Z plane, as shown in Figure \ref{fig:CS}. The scene reflectivity is denoted as $x_c(r, \tau_1, \tau_2)$ where $\tau_1 = (\sin\theta_1)/2$ and $\tau_2 = (\sin\theta_2)/2$.  At excitation wavelength $\lambda$, the 2D array outputs $y_{d_1,d_2, \lambda} = \widehat{x}_c(1/\lambda, \, 2d_1/\lambda, \, 2d_2/\lambda)$ are samples of the Fourier transform of the scene sampled in the region bounded by  $\big ( \frac{\pm D_1 }{\lambda}, \frac{\pm D_2 }{\lambda} \big )$, where $D_1$ and $D_2$ are the dimensions of the 2D array.  From now on, we use the 1D array to discuss our model for the sake of notational brevity. However, all our discussion and results hold for both cases and all simulations use 2D arrays.

We are mainly interested in a broadband excitation scenario. We assume that the excitation signal is a broadband pulse  with bandwidth in the interval  $[\lambda_{\min}, \lambda_{\max}]$. If the broadband signal used is $s_b(t)$, then the received signal at location $d$ is 

\begin{equation}
\widetilde{y}_{d}(t) =  \int\limits_{-\pi/2}^{\pi/2}\int p(r,\theta)s_b(t - r/c - d\sin\theta)\, d r\,d\theta.
\end{equation}

\begin{figure}[!tb]
\centering
\begin{minipage}[b]{0.48\textwidth}
\centering
\subfloat[]{\label{fig:wedge_a}\includegraphics[scale = 0.4]{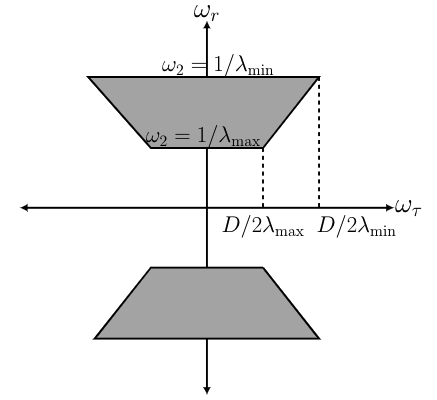}}
\end{minipage}
\hfill
\begin{minipage}[b]{0.48\textwidth}
\centering
\subfloat[]{\label{fig:wedge_b}\includegraphics[scale=0.35]{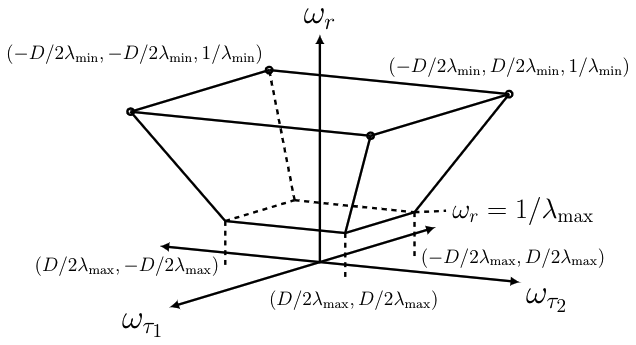}}
\end{minipage}
\caption{\small \sl  (a) shows the region in the Fourier domain of a target scene acquired by using broadband excitation and a finite 1D aperture. (b) shows the same for a 2D array imaging a 3D scene.}
\label{fig:wedge}
\end{figure}

The complex amplitudes at different wavelengths, $y_{d, \lambda}$ can then be obtained as the Fourier coefficients of  $\widehat{y}_{d}$ after taking its temporal Fourier transform. Since the lateral frequency support increases with decrease in the excitation wavelength, the accessible Fourier domain has a trapezoidal shape as illustrated in Figure \subref*{fig:wedge_a}. Similarly, for a 2D antenna, the region in the Fourier domain measured is shown in Figure \subref*{fig:wedge_b}.

\par It is now clear that a finite aperture and a finite bandwidth of excitation together offer bandlimited measurements of the target scene. If  $\mathcal{F}$ represents the continuous domain Fourier transform and $\mathcal{S}$ represents the bandlimiting operator supported only in this accessible window, the continuous aperture measurements $y_c = \{y_{d,\lambda} \}$ can be modeled as 

\begin{equation}
\label{eq:contin_repn}
y_c = \mathcal{S}\mathcal{F}x_c = \{\widehat{x}_c(\omega_\tau, \omega_r): (\omega_\tau, \omega_r) \in \mathcal{W} \}
\end{equation} where $\mathcal{W}$ represents the trapezoidal region shown in Figure \ref{fig:wedge}. Broadband imaging is hence the task of collecting Fourier measurements in a bandlimited region and inferring the target reflectivity profile using these measurements.

\subsection{Array measurement model}
\label{subsec:array_measurement_model}

\label{sec:array_model}

Our discussion above describes what can be observed through a finite aperture. In practice, we must measure this signal using a discrete array of sensors. This limits the Fourier domain measurements considered in the previous section to only a discrete set of samples. Let the 1D antenna considered in the previous section be an array of $M$ discrete antenna elements placed uniformly at coordinates $d_{\frac{-M}{2}}, \cdots, d_{\frac{M}{2}} \in [-D/2, D/2]$. For ease of explanation, we consider a discrete set of $K$ excitation wavelengths $\{ \lambda_1 = \lambda_{\min}, \lambda_2, \cdots, \lambda_K = \lambda_{\max} \}$. When the complex amplitudes are measured at these wavelengths, the Fourier samples obtained are located on a pseudopolar grid, as shown in Figure \ref{fig:ppgrid}. Measurements at these wavelengths can be obtained by using a single broadband excitation, as explained earlier. The set of all measurements $\{y_{m, \lambda} \}$ where $m = 1,\cdots,M$, $\lambda \in \{\lambda_1, \cdots, \lambda_K \}$ can be denoted by a vector  $ y \in \C^{MK}$. The set of measurements at each wavelength $\lambda$ is denoted by $y_\lambda \in \C^M$. 
\begin{figure}[!t]
\centering
\includegraphics[scale=0.22]{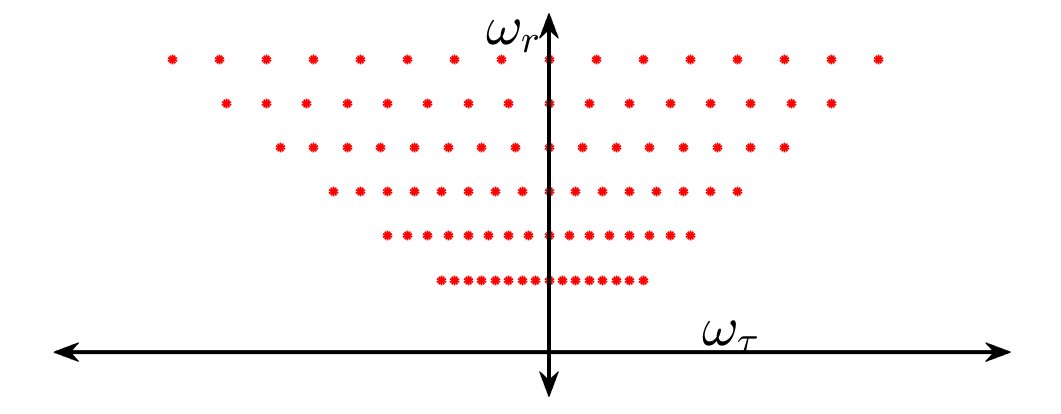}
\caption{\small \sl With a discrete array, we can only collect discrete samples in the Fourier domain of the target scene. These samples are placed uniformly along the lateral frequency axis and and at regular intervals on the vertical frequency axis. Each sample denotes the output at a single array element, at a single excitation frequency. The limits for $\omega_{\tau}$ at each $\omega_r$ depends on the aperture size and spans the region $[-\omega_r D, \ \omega_r D]$. A discrete set of temporal wavelengths is assumed instead of a continuous support.}
\label{fig:ppgrid}
\end{figure}

Traditional imaging involves collecting the samples shown in Figure \ref{fig:ppgrid} to reconstruct the target reflectivity profile. This can be achieved in many ways. One way is to directly read out the output of each antenna element. This amounts to measuring $y$ directly. Physically, this can be realized by using a broadband pulse as before and then taking a Fourier transform. Alternatively, one could cycle through a set of narrowband excitation signals (stepped frequency excitation) and collect array measurements at each wavelength. In either case, this method would require $M$ array element readouts. 

Another standard method for acquiring the measurements is to collect linear combinations of the array outputs. In this method, the output of each array element is weighted and then added to the output of other elements. In particular, the procedure known as \textit{beamforming} obtains specific linear combinations of the array outputs that induce spatial directivity. In narrowband beamforming, the weights are chosen such that the time delays for signals coming from a particular spatial direction are compensated for, hence ``focusing" the array in that physical direction, as outlined in \cite{6600891}. Instead of acquiring $M$ direct read-outs, $M$ linear combinations are acquired.  

Acquiring linear combinations at each excitation wavelength  separately is equivalent to acquiring linear combinations of samples along each row of points in Figure \ref{fig:ppgrid}. A different set of linear combinations may be used for each excitation wavelength. Mathematically, if the excitation wavelength is $\lambda$, the measurements made in time domain are 

\begin{equation}
\widetilde{z}_{i,\lambda}(t) = \sum_{m=1}^M  \phi_{\lambda}(i,m) \widehat{y}_{m,\lambda}(t)
\end{equation}  for $i = 1,\cdots, M$. The vector of complex amplitudes $z_\lambda \in \C^M$ is then given by 

\begin{equation}
z_{\lambda} = \phi_{\lambda} y_\lambda
\label{eq:nbf}
\end{equation} where $\phi_{\lambda}$ is in general an $M \times M$ well-conditioned matrix whose $(i,m)^{\text{th}}$ element is $\phi_{\lambda}(i,m)$. %and $z\in \C^{M}$ is an. 
If a single broadband pulse is used for excitation, then the set of weights for the linear combinations at different wavelengths are constrained to be the same and the vector of complex amplitudes at different wavelengths are given by:

\begin{equation}
z_{\lambda} = \phi y_\lambda
\label{eq:bbf}
\end{equation} where $\phi$ is now common across all wavelengths. In general, we refer to acquiring linear combinations of the array elements as \textit{coded aperture acquisition}. 

Yet another way of acquiring the samples in Figure \ref{fig:ppgrid} is to use coded frequency excitation where the excitation signal is a linear combination of various wavelengths. If the array elements are directly read out, then the measurements obtained can be interpreted as taking linear combinations of the samples along the radial lines in Figure \ref{fig:ppgrid}. 

A variant of aperture coding was considered in \cite{1397349} where the authors propose using different subarrays at different times and then using interpolation techniques to acquire all the samples shown in Figure \ref{fig:ppgrid}. Subsampling the array is equivalent to using binary codes on the aperture. However, their signal model is different from ours and hence the paper does not consider imaging with fewer than $M$ measurements.  

We consider a particular signal model: we assume that the target to be imaged is range-limited. With this model, we study the number of array measurements required to image the scene when broadband excitation is used. We will show that coded apertures (\eqref{eq:nbf} and \eqref{eq:bbf}) can be used with broadband excitation to achieve highly efficient imaging of range-limited target scenes. Unlike standard methods that require $M$ direct read-outs or $M$ linear combinations of the array outputs, we will use far fewer ``generic" linear combinations to image range-limited scenes. In other words, we show that when the target scene is range-limited, the coding matrices $\phi$ (or $\phi_{\lambda}$) can be highly underdetermined, without assuming any sparsity in the scene.

\section{Signal model and degrees of freedom}
\label{sec:sig_model}

In this section, we introduce our signal model and make some observations about the model. These observations, a signal acquisition scheme based on them, and theoretical guarantees on the proposed signal acquisition method are our main contributions.

\subsection{Array measurements of range-limited scenes}
\label{subsec: array_measurements_RL_scenes}
Our signal model is that of a range-limited scene.  A finite bandwidth and aperture allow us to observe only a part of the Fourier transform of the image. A finite range restricts the number of degrees of freedom of this observed region of the Fourier transform. We intend to take advantage of this to achieve a more efficient sampling of the bandlimited spectrum of the image and achieve hence faster imaging.

Briefly switching back to the setting of a continuous aperture, we can obtain an analog of \eqref{eq:contin_repn} for range-limited targets. If $\mathcal{Q}$ represents the self-adjoint operator that truncates an image to a range limit $R$, the continuous aperture measurements can be modeled as

\begin{equation}
\label{eq:cont_repn_rl}
y_c = \mathcal{S}\mathcal{F}\mathcal{Q}x.
\end{equation} The discrete array measurements can be modeled as

\begin{align}
	y_{m}(\lambda) & =  \int\limits_{-\pi/2}^{\pi/2}\int\limits_{R} p(r,\theta)\e^{-\j 2\pi (r +{d_m\sin\theta })/\lambda}\, d r\,\d\theta.
\end{align} 

We now demonstrate the effect of range-limitedness using a target scene with \textit{delta} thickness present at a constant known range: where the scene has only one reflector per angle, with each reflector present at a constant known distance from the array. The underlying effect on the Fourier domain samples extends to scenes with a more general range limit. 

Consider a scene with delta thickness at a constant range $R_0$ from the antenna array. The time domain outputs at different wavelengths can be modeled using \eqref{eq:narrowband_eqn} as:

\begin{align}
\label{eq:CR_cont_repn}
\widetilde{y}_{m,\lambda}(t) = \e^{\j 2\pi (ct+R_0)/\lambda} \! \int\limits_{-1/2}^{1/2} p(R_0,\tau)\e^{-\j 2\pi \omega_{\tau}\tau}
 \,\d\tau.
\end{align}

Considering just the amplitude as before, we have
\[ 
y_m({\lambda}) = \int\limits_{-1/2}^{1/2} p(R_0,\tau)\e^{-\j 2\pi 2d_m \tau/\lambda}\,\d\tau.
\]
Now define 
\[ 
g(\omega_{\tau}) = \int\limits_{-1/2}^{1/2} p(R_0,\tau)\e^{-\j 2\pi \omega_{\tau} \tau}\,\d\tau.
\]
The array outputs $y_{m}( \lambda)$ are then just \textbf{ samples of the same function} $g(\omega_{\tau})$ sampled uniformly in $[-D/\lambda,D/\lambda]$ (modulo known  scaling factors). This is illustrated in Figure \ref{fig:ppgridCR}. In the 2D array case, the ``slices" of the trapezoid corresponding to different excitation wavelengths sample a common function. As the range limit increases but remains finite, the functions sampled at different wavelengths start to differ, but still have limited degrees of freedom.

\begin{figure}[!tb]
\centering
\includegraphics[scale=0.27]{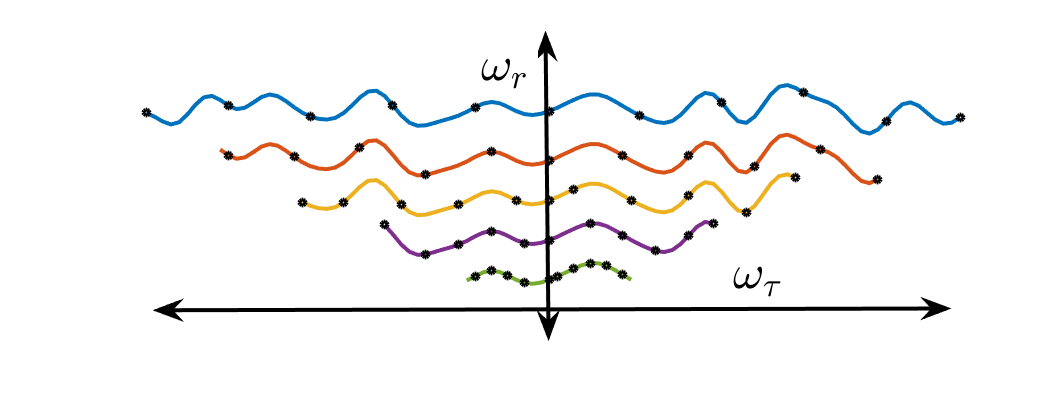}
\caption{\small \sl For an image at a constant range, the array measurements at different wavelengths are samples of a common function}
\label{fig:ppgridCR}
\end{figure} 
 
From the above discussion, it is clear that for scenes with delta thickness, collecting the full set of samples at the lowest excitation wavelength provides all  available information. Collecting further samples at higher wavelengths offers no advantage. However, this redundancy can be used to collect fewer spatial samples but with broadband excitation. Similarly, collecting the full set of broadband measurements for scenes with higher but finite range limits results in a number of measurements greater than the number of degrees of freedom. It is thus natural to expect that the target scene can be reconstructed with a number of measurements $l \ll M$, owing to the limited number of degrees of freedom.

For computational purposes, we can discretize the target scene and the array imaging operator. For a scene with delta thickness at a depth $R_0$, let $x_{R_0} \in \R^N$ denote the target scene $x_c$ sampled uniformly with $N \geq M$. Similarly, the integral mapping the target scene to the array measurements at wavelength $\lambda$ can be discretized as a matrix $A_{R_0,\lambda} \in \C^{M \times N}$ operating on $x_{R_0}$:

\begin{equation}
A_{R_0,\lambda}(m,n) = \exp^{-j2\pi R_0/ \lambda}\exp^{-i2\pi 2d_m (-0.5+ n/N)/\lambda}.
\label{eq:A_r0}
\end{equation}With this notation in place, the array measurements can be expressed as

\[ y_{R_0,\lambda} = A_{R_0,\lambda} x_{R_0}. \]

Similarly, a more general scene with a range limit $R$ can be discretized as a vector of reflectivities as $x_R \in \R^{N \times D}$ where $N$ represents the number of discrete samples along the $\tau$ axis and $D$ the number of samples along the $r$ axis. Let the scene lie between the range limits $R_{\min}$ and $R_{\max}$. Define $d_r = \lfloor n/N + 1 \rfloor, \ n_{\tau} = \mod(n,N)$. Then the discretized array operator can be expressed as a matrix $A_R \in \C^{M \times ND}$ as

\begin{equation}
A_{R,\lambda}(m,n) = \exp^{-j2\pi (2d_m (-0.5+ n_{\tau}/N) + (R_{\min} + d_rR/D))/\lambda}
\label{eq:A_R}
\end{equation} 
and the array outputs can be expressed as 
\[y_{R,\lambda} = A_{R,\lambda}x_R. \]

Since our signal model considers only range-limited scenes, we drop the subscripts $R$ and $R_0$ for further discussion. We will denote antenna array measurements at excitation wavelength $\lambda$ as $y_{\lambda} \in \C^{M}$, the discretized target scene as vector $x_0 \in \R^{ND}$ (or $\R^N$ for scenes with delta thickness) and the array imaging operator as $A_{\lambda} \in \C^{M \times ND}$ (or $\C^{M \times N}$ for scenes with delta thickness). $A_{\lambda}$ will incorporate knowledge of the target range profile, which is assumed to be known a priori. This helps us focus on the advantage of range-limitedness in itself. We later show how unknown range profiles can be handled algorithmically.  

When the measurements at all the $K$ wavelengths are considered, we obtain the linear system

\begin{equation}
\label{eq:broadband_discrete}
y = \begin{bmatrix}
y_{\lambda_1} \\
y_{\lambda_2} \\
\vdots\\
y_{\lambda_K}
\end{bmatrix} = \begin{bmatrix}
A_{\lambda_1}\\
A_{\lambda_2}\\
\vdots \\
A_{\lambda_K}
\end{bmatrix}x_0 = Ax_0.
\end{equation} The collection of array outputs at all the wavelengths lie in the column space of $A$. The effective dimension of this subspace determines the number of degrees of freedom in $y$. 
When the measurements are acquired with a coded aperture, we can model the measurements as
\begin{equation}
\label{eq:ap_coded_measurements}
z = \Phi y =  \begin{bmatrix}
\phi & 0 & \cdots & 0\\
0 & \phi & \cdots & 0 \\
\vdots & \vdots & \ddots & \vdots \\
0 & 0 & \cdots & \phi
\end{bmatrix}  \begin{bmatrix}
y_{\lambda_1} \\
y_{\lambda_2} \\
\vdots\\
y_{\lambda_K}
\end{bmatrix}  = \Phi Ax_0.
\end{equation}
Our main goal is to show that the target scene can be constructed using highly underdetermined matrices $\Phi$ with no loss in resolution compared to the reconstruction obtained using the full set of measurements $y$.

\subsection{Degrees of freedom of range-limited scenes}

\begin{figure}[!tb]
\centering
\includegraphics[scale=0.22]{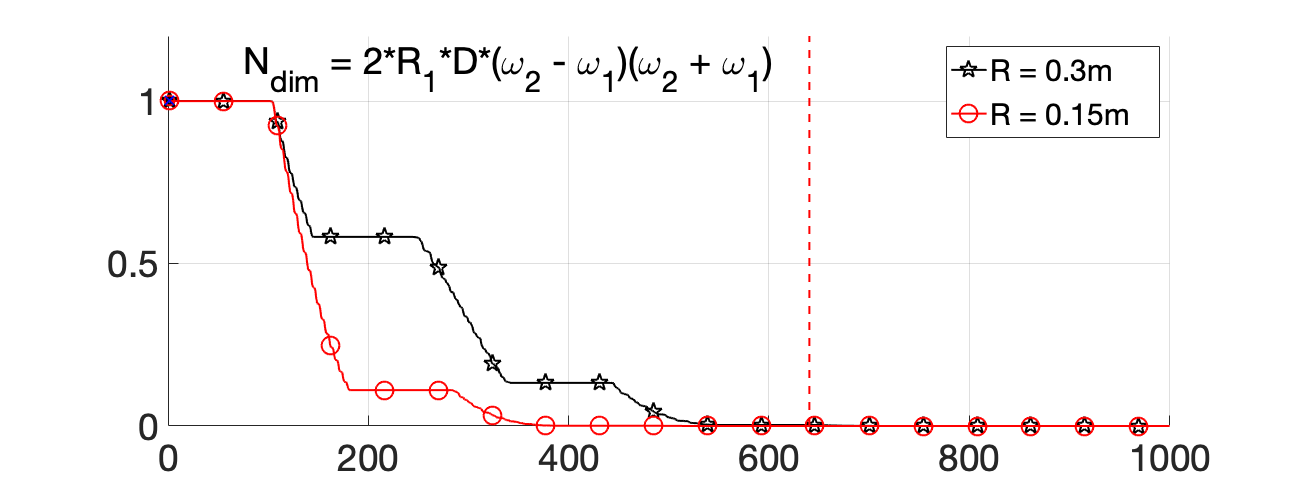}
\caption{\small \sl Eigenvalue decomposition of operator $A$ in the 1D array case for a range limit of $R$, aperture size $D$ and excitation frequency range of $[\omega_1/c, \ \omega_2/c]$. $R1 > R2 > R3$}
\label{fig:svd_multiR}
\end{figure}

In \ref{subsec: array_measurements_RL_scenes} we observed that bandlimited samples of range-limited target scenes have limited number of degrees of freedom. This is a consequence of simultaneous concentration of energy in both spatial and spectral domains. Signals with such a property can be well approximated by a number of basis functions that is proportional to the product of the area/volume of the spatial and spectral supports. This has been studied in a set of seminal papers by Slepian, Landau and Pollock \cite{1454379, BLTJ:BLTJ3976, BLTJ:BLTJ3977, BLTJ:BLTJ1037, BLTJ:BLTJ2104, BLTJ:BLTJ3279} and by Simons et.al., in \cite{Simons2010, Simons2011}. In essence, the range space of space-limiting, band-limiting operators such as that in \eqref{eq:cont_repn_rl}  is approximately finite dimensional. This approximate dimension is the number of effective degrees of freedom of signals well concentrated in spatial and spectral domains. An efficient basis for the representation of such signals is the \textit{prolate} basis \cite{1454379}.  Smaller the spatial and spectral supports, smaller is the number of degrees of freedom. See~\cite{KZWRD_FST_2019,ZKWDR_Roast_2018,ZKDRW_eigenvalue_2018} for a quantitative non-asymptotic characterization of these properties in the discrete case.

Since the array imaging operator is an approximation of the continuous domain model in \eqref{eq:cont_repn_rl}, its range space has a low dimensional structure. To illustrate this, we present the singular values of $A$ for target images with a finite range limit in Figure \ref{fig:svd_multiR}. The product of spatial and spectral supports in this case is approximately $2RD(\omega_2 - \omega_1)(\omega_2 + \omega_1)$. For these plots, an array with $213$ elements was used and samples were collected at $25$ wavelengths placed uniformly between $2$GHz and $4$GHz. The total number of samples collected is hence $5325$, but these samples lie in a subspace of dimension approximately only $640$ or $320$, for the two range limits considered.  For the same antenna array, we show the singular values of the operators associated with target scenes that have delta thickness in Figure \ref{fig:svd_delta}. As expected, such scenes lie subspaces of even smaller dimensions. 

\begin{figure}[!tb]
\centering
\includegraphics[scale=0.22]{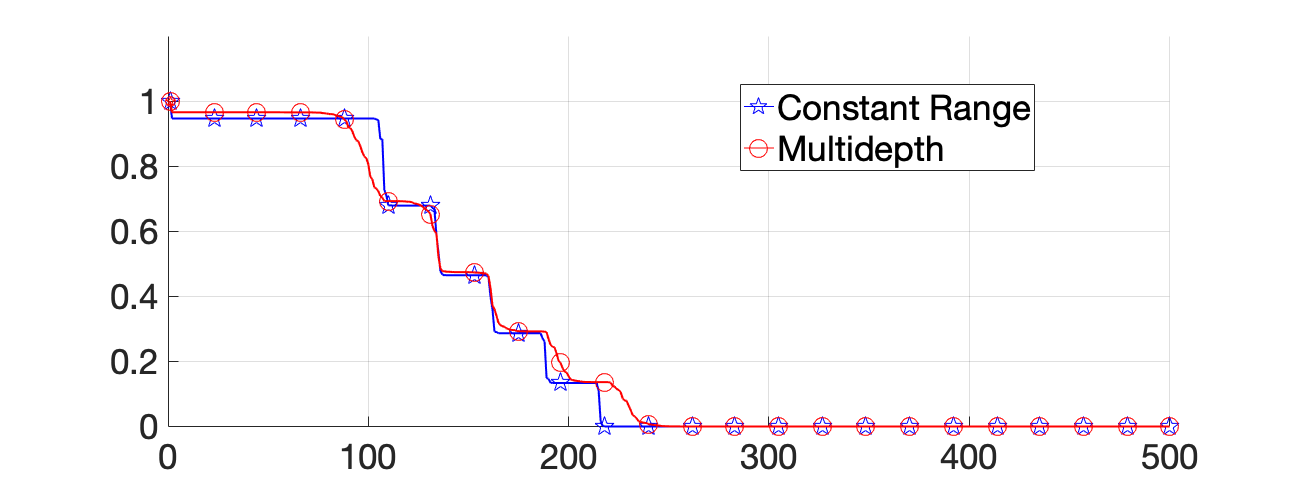}
\caption{\small \sl Spectra of the operators for scenes with delta thickness. `Constant range' describes a scene where the reflectors at all angles are at a constant depth from the array center and `Multi-depth' describes a scene in which each reflector is at a different depth from the array. In each case, the infinite dimensional continuous domain image can be efficiently represented using a subspace of relatively small dimension}
\label{fig:svd_delta}
\end{figure}

\begin{figure}[!tb]
\centering
\includegraphics[scale=0.35]{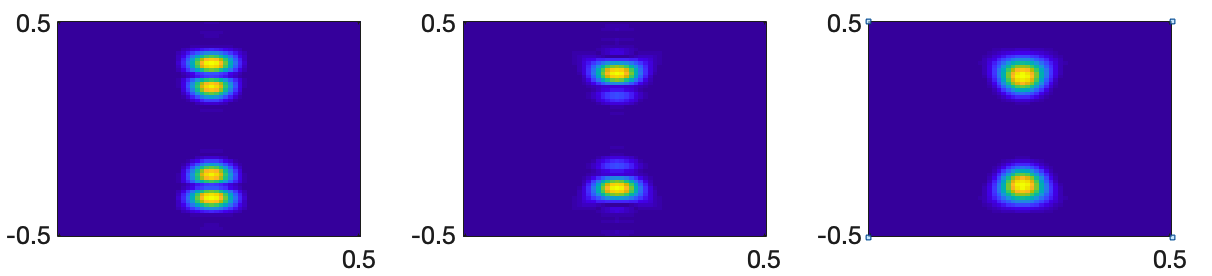}
\caption{\small \sl Spectra of the top eigenvectors of $A^TA$. The eigenvectors have a frequency support that is mostly concentrated in the trapezoidal area and a spatial support limited to $R$ as defined by the model.}
\label{fig:spectra}
\end{figure}

A special case of range-limitedness is when the target has delta thickness. Figure \ref{fig:svd_delta} shows the singular value decomposition of the array operator for two examples of such scenes: one where all the reflectors are at a constant known distance from the array center; the other where each reflector is at a known but different distance from the array center.  We also show the Fourier transforms of the eigenvectors of $A^TA$ associated with the most and least significant eigenvalues in Figure \ref{fig:spectra}.

The low dimensional structure in the array measurements forms the basis of our proposed imaging method. We show that this structure enables imaging with very few array measurements. In the following sections, we set up the reconstruction problem and provide theoretical guarantees for image reconstruction.

\section{Coded aperture image reconstruction}
\label{sec:coded ap acq}

\subsection{Image reconstruction method}
\label{subsec:img_recon}
The limited number of degrees of freedom of broadband measurements leads us to ask the following question: what is the number of broadband measurements required to image range-limited target scenes? In order to answer that question, we first set up our reconstruction method. We do not assume any structure on the target scene such as sparsity, low total variation norm, or other structure apart from it being range-limited. Hence, with the measurement model used in \eqref{eq:ap_coded_measurements}, we solve the ordinary least squares problem
\begin{equation}
\label{eq:ap_coded_reconstruction}
\underset{x}{\argmin}{\|\Phi y - \Phi Ax \|_2^2}.
\end{equation}
Our standard of comparison is the ordinary least squares estimate obtained using the full data:
\begin{equation}
\label{eq:OLS}
\underset{x}{\argmin}\|y - Ax \|_2^2.
\end{equation}
The theoretical question we answer is: If each diagonal block $\phi$ in \eqref{eq:ap_coded_measurements} is of size $l \times M$, what is the sample complexity $l$ that can achieve the same reconstruction result as \eqref{eq:OLS} using \eqref{eq:ap_coded_reconstruction}? 

\subsection{Sketched least squares systems}
\label{subsec:sketching}
Problems of the type \eqref{eq:ap_coded_reconstruction} are the subject of a recently popular area of research in the numerical linear algebra  community, known as matrix sketching \cite{woodruff2014sketching}. Sketching techniques use random projections of the rows (or columns) of a matrix $A$ to obtain approximations to the matrix. The new matrix obtained has much smaller dimensions and is used to solve the original linear algebra problem. Sketching can be applied to a multitude of problems such a linear regression, low rank matrix factorization, and subspace clustering among others.

\par 
If $y$ denotes measurements of an unknown signal $x$ observed through the linear operator $A$, then the least squares solution is given by \eqref{eq:OLS}. When $A$ is low rank ($\text{rank}(A) \ll \min(\tilde{M},N)$), $x$ can be solved for using fewer measurements. This forms the fundamental basis of the idea of sketching. Formally, if $S$ is a compressive matrix, sketching solves the following problem:
\begin{equation}
\label{eq:sketch}
x_{\text{SLS}} = \text{arg} \min_x \|Sy - SAx \|^2,
\end{equation}
where the subscript SLS stands for sketched least squares.
The key point of course is that $S A$ is now a much smaller matrix and thus the solution can be computed quickly. This technique has generated a lot of interest in the context of large scale problems where the ambient dimensions are much bigger than the inherent dimensionality. In a sensing system, this translates to reconstructing a target signal from compressed measurements, as common in compressed sensing systems. However, we point out that unlike standard compressed sensing setup, we do not assume sparsity or compressibility in the signal $x$. Hence, we aim to solve \eqref{eq:ap_coded_reconstruction} without  imposing any structure on $x$. 

Our proposed system can be modeled as a sketched linear system, but with a  particularly structured sketching matrix:
\begin{equation}
\label{eq1}
\Phi y = \Phi A x_0 = Yx_0,
\end{equation}
where $\Phi$ is a block diagonal matrix with $k$ repeated blocks
\begin{equation}
\label{eq:RBD}
\Phi = 
\begin{bmatrix}
\phi & 0 & \cdots & 0\\
0 & \phi & \cdots & 0\\
\vdots & \vdots & \ddots & \vdots \\
0 & 0 & \cdots & \phi
\end{bmatrix}.
\end{equation} We refer to such a matrix as a repeated block diagonal (RBD) matrix. 

\subsection{Linear algebraic interpretation of range-limitedness}

The least squares program in \eqref{eq:OLS} searches for an image that best explains the measurements obtained, and lies in the row space of $A$. The row space of $A$ is not only approximately low dimensional for range-limited images, but has a second tier of structure that makes aperture coding a very efficient way of obtaining fewer measurements. In particular, the relationship between the subspaces associated with the linear operators at various excitation wavelengths for range-limited images determines the number of aperture codes required. 

Consider array imaging at a single wavelength $A_{\lambda_i}$. The least squares estimate in this case looks for a target in the row space of $A_{\lambda_i}$ that best explains the measurements. We denote this subspace as $\mathcal{S}_i$. Similarly, imaging using wideband excitation results in an estimate that lies in the union of subspaces $\mathcal{S}_1, \mathcal{S}_2,  \cdots, \mathcal{S}_k$. The sample complexity of aperture coding depends highly on the relationship between these subspaces. 

Let $A$ now denote a general $km \times n$ matrix of a finite rank $r$ such that $A = \begin{bmatrix}
A_1^T & A_2^T & \cdots & A_k^T
\end{bmatrix}^T$ where each $A_i$ is $m \times n$. Let the rank of $A_i$ be $r_i$. Without loss of generality, we can assume that $r_i \geq r_j$ for $i > j$, since the ordering of the row groups does not matter. We can then obtain the following factorization:
\begin{equation}
\label{exp}
\begin{bmatrix}
A_1\\
A_2\\
\cdot\\
A_k
\end{bmatrix} =  \begin{bmatrix}
C_{11} & 0 & \cdots & 0\\
C_{21} & C_{22} & \cdots &0\\
\vdots & \vdots & \ddots & \vdots\\
C_{k1} & C_{k2} & \cdots & C_{kk}
\end{bmatrix}\begin{bmatrix}
V_1^T\\
V_2^T\\
\vdots\\
V_k^T
\end{bmatrix} = CV^T
\end{equation} where $C_{ij} \in \R ^{m \times d_j}$, each $C_{ii}$ is full column rank when $d_i \neq 0$, and $V$ is an $r \times n$ orthonormal matrix. The factorization is such that the row space of $A_1$ is the span of the orthobasis $V_1$, the row space of $A_2$ is included in the span of $V_1$ and $V_2$. In general, $[V_1 \ V_{2} \ \cdots V_i]$ includes an orthobasis for the row space of $A_i$. This factorization is equivalent to a block QR factorization of $A^T$ and can be obtained for any general matrix A. We are particularly interested in the $d_i$'s, as they capture the relationship between various subspaces. The diagonal blocks $C_{ii}$'s represent the energy in the subspace orthogonal to the union of the row spaces of the previous blocks $A_{1}, \cdots, A_{i-1}$. Hence, smaller values of $d_i$ indicate that the subspaces have a high degree of overlap.

A special case of high overlap is when the row spaces have a nested structure: $\row(A_{i-1}) \subset \row(A_{i}) \ \forall i = 2, \cdots, K$. In this case, the off-diagonal blocks $C_{ij}, i < j$ and the orthogonal blocks $V_{1}, \cdots, V_{i-1}$ capture a significant part of $\row(A_i)$. For low rank systems, this naturally leads to smaller values of $d_i$. In contrast, when the row spaces are all almost orthogonal, the $d_i$'s are all large and the off-diagonal blocks $C_{ij}, i \neq j \approx 0$. We will later show that the broadband array imaging operator has the nested subspace structure for certain range-limited scenes.
  
 Let us now relate the above factorization to the context of imaging. To begin, suppose that we use wavelengths up to $\lambda_i$. Then, $d_{i+1}$ represents the rank of the update required to incorporate information from a new, lower wavelength $\lambda_{i+1}$. It represents the \textit{innovation} added by the measurements at the new wavelength. 
 
 As the range limit decreases, the overlap between the subspaces $\mathcal{S}_i$ increases, increasing the redundancy across wavelengths. This leads to the possibility of higher subsampling rates in the physical array domain. In the limiting case of only one reflector per angle, the subspaces have a nested structure. As we will observe in Section \ref{sec:signalRecovery}, this plays a crucial role in determining the number of aperture codes required for successful imaging. 

Our theoretical results formally state the effect of the relationship between the subspaces $S_i$ on the number of aperture codes required for imaging. Theorem \ref{thm:main} provides a non-trivial estimate of the number of
measurements needed for a given excitation bandwidth. Theorem \ref{thm:manyK} provides conditions under which a given set of excitation wavelengths $\lambda_1, \cdots, \lambda_K$ allow the number of spatial measurements to be reduced by a factor of $k$. Using these conditions and a given bandwidth of excitation, one can choose the set of excitation wavelengths and very few coded measurements to achieve imaging with no loss in resolution.

\section{Signal recovery from aperture coded measurements}
\label{sec:signalRecovery}

In the previous section, we set up the aperture coding problem as a sketched least squares problem that has a particular structure dictated by the physical problem of array imaging. In this section, we derive mathematical guarantees for such sketched systems and provide estimates of the required sample complexity. 

We start by reviewing the conditions that any general sketching operator has to satisfy in order for the solution to the sketched least squares problem $x_{\text{SLS}}$ to be close to the solution of the original ordinary least squares solution $x_{\text{LS}}$. In the noiseless case
\begin{equation}
\label{eq:VVT}
x_{\text{LS}} = A^{\dagger}y = A^{\dagger}Ax_0 = VV^*x_0
\end{equation} 
where $A = U\Sigma V^*$ is the SVD of the linear operator and $A^\dagger$ denotes the pseudoinverse of $A$. \eqref{eq:VVT} shows that the least squares solution is a projection of the true solution onto the row space of $A$.  Hence, any sketching operator $\Phi$ should preserve the row space of $A$. It has been well established in literature that a number of random projections of the rows of $A$ greater than or equal to its rank capture the row space in case of exactly low rank matrices \cite{Nhalko}.  When $A$ of size $m \times n$ and $\text{rank}(A) = r \ll \min(m,n)$, if $\Phi$ is a $l \times m$ dense standard normal random matrix with $l \geq r$, then $\row(A) \subset \row(\Phi A)$.
Since by construction we also have $\row(\Phi A) \subset \row(A)$, we have
\[ V_{\Phi A} V_{\Phi A}^* - VV^* = 0. \] Let 
\[ {x}_{\text{SLS}} = Y^\dagger \phi y = V_{\Phi A}V_{\Phi A}^*x_0.\] The least squares estimate from the sketched measurements $\phi y$ is same as that from the full observation $y$, since \[\|x_{\text{LS}} - {x}_{\text{SLS}} \| \leq \| VV^* - V_{\Phi A}V_{\Phi A}^* \|\|x_0 \| = 0 .\] This idea forms the basis of using sketched measurements to solve a least squares problem. For any sketching matrix $\phi$, a necessary and sufficient condition in the noiseless case is $\|(I - P_{(\phi A)^*})A^*  \| = 0$.

Our goal is to replicate the result that $l \geq r$ is sufficient, but for an RBD matrix. For such a matrix (shown in  \eqref{eq:RBD}), the equivalent result would be that a total number of measurements $Kl \geq r$ suffice to capture the row space of the matrix $A$. However, due to the highly structured nature of a block diagonal matrix, such a result does not hold uniformly for all matrices $A$. We analyze the conditions on $A$ under which such a result holds and show that array imaging matrices do obey these conditions, thus allowing for spatial subsampling. 

\subsection{Least squares with a block diagonal sketching matrix: General case}

RBD matrices obtain \textit{localized} random projections: they take linear combinations of only a subset of the rows. In this section, we provide guarantees on the the error $\|(I - P_{(\Phi A)^*})A^*\|$ when $\Phi$ is an RBD matrix. The focus will be on the sample complexity $l$ required to drive this error to $0$ with high probability.  

It is immediately clear how to achieve this when we take $\displaystyle l \geq \max_i  \text{rank}(A_i)$. Let  $A = [A_1^T \ A_2^T \cdots \ A_K^T]^T$ where each $A_i$ is of size $m \times n$ and has rank $r_i$. Let $Y = \Phi A$. Hence 
\begin{equation*}
Y = \begin{bmatrix} \phi A_K \\ \phi A_{K-1} \\ \vdots \\ \phi A_1 \end{bmatrix} = \begin{bmatrix} Y_k \\ Y_{k-1} \\ \vdots \\ Y_1 \end{bmatrix}.
\end{equation*} 
Since 
\begin{equation}
\|(I - P_{Y^*}){A^*}\| \leq \sum_{i = 1}^k \|(I - P_{Y_i^*})A_i^*\|, \\
\end{equation} and 
$\|(I - P_{Y_i^*})A_i^*\| = 0$ for $l \geq r_i$, we obbtain
\begin{equation}
\|(I - P_{Y^*})A^* \| = 0
\end{equation}
for $l \geq \max_i r_i$.  

Compared to using a dense random matrix, this can be worse by a factor of $K$. Intuitively, this straightforward application of results from \cite{Nhalko} leads to capturing of the subspace spanned by each group of rows $A_i$ individually, without considering the overlap between the subspaces.  This is addressed in our first analytical result (Theorem \ref{thm:main}), which provides a simple but non-trivial estimate of the number of random projections required. We then improve this result in Theorems \ref{thm:final_k2} and \ref{thm:manyK}. 

\begin{theorem}
\label{thm:main}
For a given matrix $A$ of size $Km \times n$, let the $d_i$ be defined as in \eqref{exp}. Let $\Phi$ be a block diagonal matrix with repeated diagonal block $\phi$ of size $l \times m$ and whose entries are chosen i.i.d.\ from the standard normal distribution. Let $Y = \Phi A$. Define $d_0 = \max_i d_i$. For $l \geq d_0$, $\|(I - P_{Y^*})A^* \| = 0$ with probability $1$.
\end{theorem}

\begin{proof}
We have 
\[\Phi A = \Phi C V^T  =  \begin{bmatrix}
\phi C_{11} & 0 & \cdots & 0\\
\phi C_{21} & \phi C_{22} & \cdots &0\\
\vdots & \vdots & \ddots & \vdots\\
\phi C_{k1} & \phi C_{k2} & \cdots & \phi C_{kk}
\end{bmatrix}\begin{bmatrix}
V_1^T\\
V_2^T\\
\vdots\\
V_k^T
\end{bmatrix}. \] Then, if each diagonal block is full rank, the matrix $\Phi C$ and hence $\Phi CV^T$ is full rank, since $V$ is just an orthonormal matrix. Consider each diagonal block of $\Phi C$, $\phi C_{ii}$. Since $l \geq d_0 $, $\operatorname{rank}(\phi C_{ii}) = d_i$ with probability $1$. Since the rank of any block triangular matrix is at least the sum of ranks of the diagonal blocks, $\operatorname{rank}(\Phi C) = \operatorname{rank}(\Phi C V^T) = \sum_i d_i = r$.  Since $\row(\Phi CV^T) \subset \row(CV^T)$ and  $\operatorname{rank}(\Phi CV^T) = \operatorname{rank}(CV^T)$, $\row(\Phi CV^T) = \row(CV^T)$. Hence, $P_{Y^*} = P_{A^*}$ and the conclusion follows.  
\end{proof}

To understand this theorem, we can first consider the intuition behind random projections: owing to the randomness of the projections, a linearly independent set of vectors in the row/column space of the original matrix is obtained and hence the subspace is captured. With an RBD sketching matrix, one can only obtain random projections within the subspace of each group of rows. However, if the subspaces spanned by these row groups overlap, we may obtain a linearly independent set of vectors that capture the row space of the whole matrix by using a few random projections of each row group. The factorization in \eqref{exp} captures this  dependence between the subspaces spanned by the row groups.

The bound in Theorem \ref{thm:main} can be tight or loose depending on the candidate matrix $A$. The bound is tight when the row spaces span non-overlapping subspaces, since in this case, random projections within a row group do not provide any information about the other row groups. When the subspaces do overlap the bound can be significantly improved, even though it is already better than the trivial bound $l \geq \max_i r_i$. For example, consider a scenario where the subspaces spanned by the $A_i$'s are nested in the subspaces spanned by the $A_j$'s for all $i < j$. In such a case, by obtaining random projections within the span of $A_i$, we are also already obtaining random projections in the subspace spanned by $A_j$ for all $i < j$. This allows for $l$ to be much smaller, even if some of the $d_i$ are large, contrary to what Theorem \ref{thm:main} predicts.  As we note in the next subsection, this is precisely the case for the linear operator associated with range-limited images.
 
 \par Although our result concerns exactly low-rank matrices, it can be generalized to matrices that are approximately low rank using perturbation theory for projection matrices \cite{Li2013}. In general, for matrices with numerical ranks much smaller than the ambient dimensions, the subspaces spanned by the $A_i$'s will overlap. This overlap reduces the number of random vectors needed per block.

\begin{figure*}[!tb]
\centering
\begin{minipage}{0.48\textwidth}
\includegraphics[width=\linewidth]{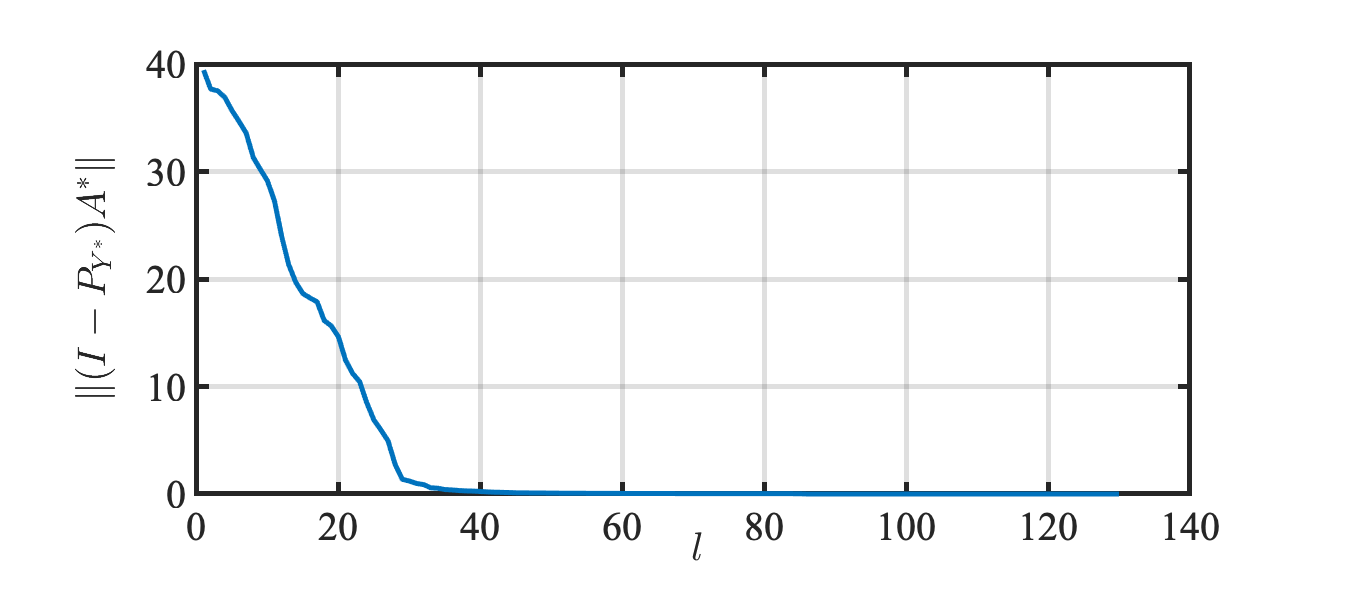}
\end{minipage}
%\hspace*{\fill}
\quad
\begin{minipage}{0.48\textwidth}
\includegraphics[width=\linewidth]{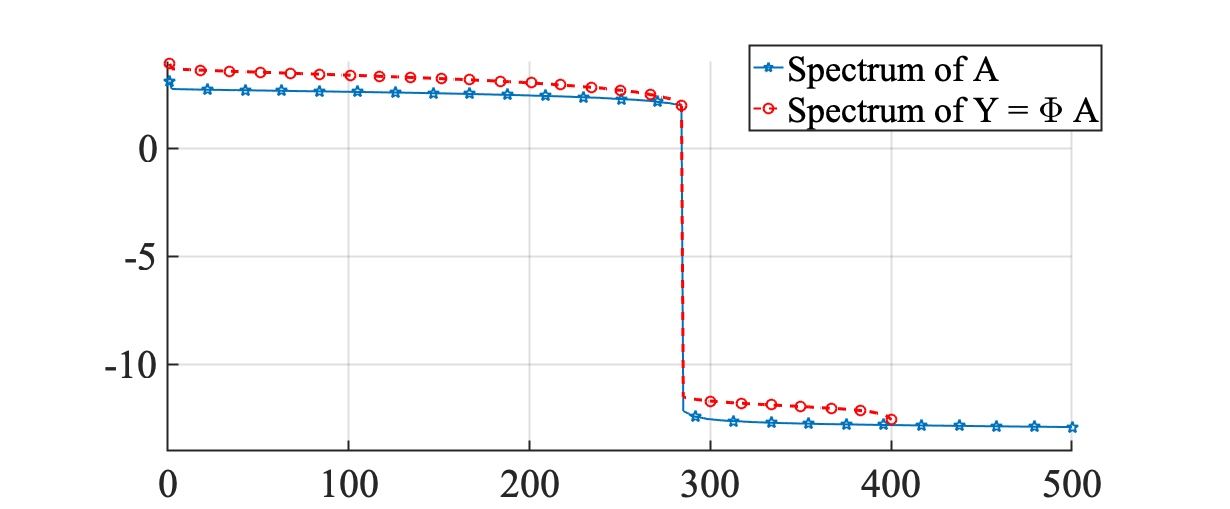}
\label{spectra}
\end{minipage}
\caption{\small \sl (a) shows the error $\|(I - P_{Y^*})A^* \|$ and (b) shows logarithm of the spectrum of $A$ and $Y$ for $l = 40$. The number of significant singular values of $A$ and $Y$ can be seen to be approximately the same and hence they have approximately the same rank and row space.}\label{error}
\end{figure*}

Figure \ref{error} shows the error $\|(I - P_{Y^*})A^* \|$ and the spectra of the matrices $A$ and $\Phi A$ for a randomly generated test matrix of size $2000 \times 1000$, with $\max_i(r_i) \approx 300$ and $d_0  = 110$.  We observe empirically that the row space was captured with $l \approx 40$. 

\subsection{Least squares with a block diagonal sketching matrix: Overlapping subspaces}

Consider a 1D array and a 2D image. For an image with delta thickness with reflectors at a known fixed depth $R_0$, the rows of $A_{\lambda_i}$ are unit-length sinusoids in the frequency range $[-D/\lambda_i, \ D/\lambda_i ]$, with a modulation term $e^{-j 2\pi R_0/\lambda_i}$. The row space of $A_i$ is therefore well approximated by the first $2D/\lambda_i$ discrete prolate spheroidal sequences (DPSS) \cite{BLTJ:BLTJ3976}. As the wavelength progresses from the highest $\lambda_{\max}$ to the lowest $\lambda_{\min}$, the subspace spanned by the rows of $A_{\lambda_i}$ is nested in the subspace spanned by $A_{\lambda_j}$ for all $j > i$. This is shown in Figure \ref{fig:nested} where the full operator with all wavelengths has a row space of the same dimension as the operator at only the highest frequency. 

\begin{figure}[!tb]
\centering
\includegraphics[scale=0.24]{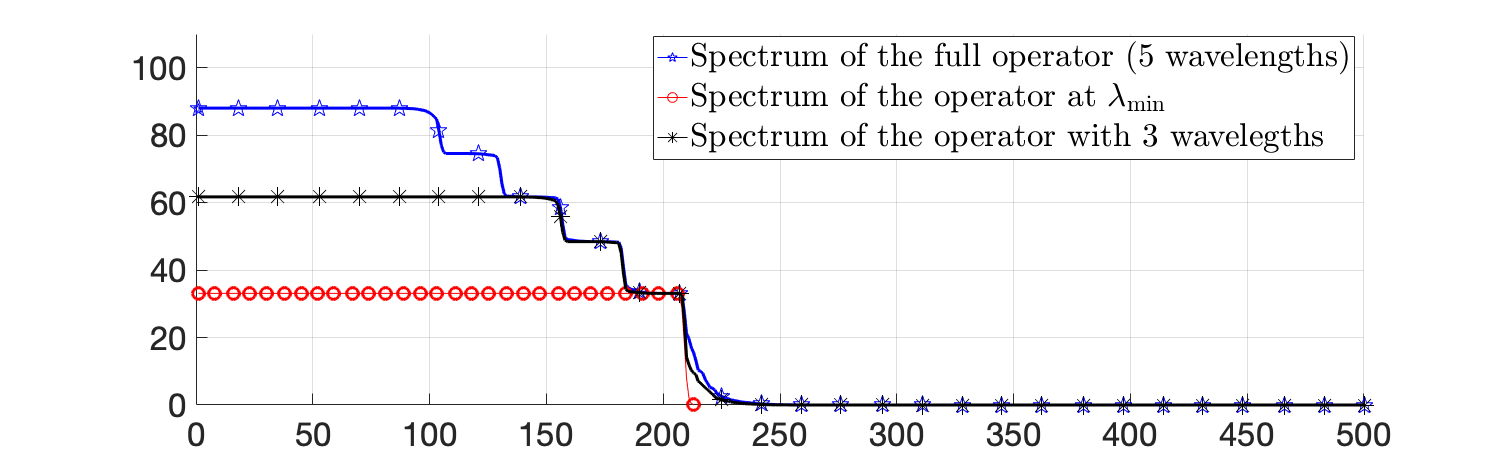}
\caption{\small \sl Normalized spectra of the full operator with multiple excitation wavelengths and of the operator at just the highest frequency. It is clear that the row spaces of the operators at the higher wavelengths are nested in that of the operator at the lowest wavelength. The same relationship is observed between the operator at any other excitation wavelength and higher wavelengths.}
\label{fig:nested}
\end{figure}
Although Theorem \ref{thm:main} is appealing because of its simplicity, it can lead to suboptimal estimates of $l$. For example, if $\lambda_{\min} = 7.5$cm and $\lambda_{\max} = 15$cm, one could only hope for a reduction in the number of measurements by a factor of $2$. But as we have already seen in Figure \ref{fig:CRmain}, we can easily achieve better subsampling.

The following two theorems provide much stronger results: they directly address the question of when $l \geq r/K$ random projections suffice to achieve $\|(I - P_{Y^T})A^T \| = 0$. In the context of array imaging, they address the question of when using $K$ different excitation wavelengths can offer the luxury of imaging with roughly only $M/K$ measurements. Their proofs are deferred to the appendix. 

\begin{theorem}
\label{thm:final_k2}
Let $A =  \bigl[ \begin{smallmatrix}A_1 \\  A_2 \end{smallmatrix} \bigr ]$, with $\row(A_1) \subseteq \row(A_2)$. Assume $A_2$ is full row rank and define $U = A_1A_2^T$. If the entries of $\phi$ are drawn from a continuous distribution, with probability 1, $M = \bigl [ \begin{smallmatrix}\phi A_1 \\ \phi A_2 \end{smallmatrix} \bigr ]$ is full row rank for any $l \leq n/2$ if and only if no real eigenvalue of $U$ has an algebraic multiplicity greater than  $n/2$. Consequently, for $l \geq n/2$, $\|(I - P_{M^T})A^T = 0\|$ with probability 1.
\end{theorem}
 
Theorem \ref{thm:final_k2} provides the necessary and sufficient conditions on an ensemble of two matrices with a nested subspace structure under which only $r/2$ projections of each block are sufficient to capture the row space. The condition prohibits the existence of an invariant subspace of dimension greater the $n/2$ that is common to both $A_1$ and $A_2$. Otherwise, $\phi A_1$ and $\phi A_2$ will both have a component along this subspace for some $l < n/2$, resulting in a loss of linear independence.

The eigenvalue distribution for the array imaging matrices is shown in Figure \subref*{fig:eigU} when $\lambda_{\min} = 7.5$cm and $\lambda_{\max} = 10$cm. It is clear that the matrices meet the required condition and hence only $M/2$ measurements suffice, in contrast to what is predicted by Theorem \ref{thm:main}, which would be at least $3M/4$.

% After the dimension of the invariant subspace is greater than $n/2$, the larger the dimesion, the more the number of random projections necessary to capture the subspace. 

In our next result, we extend the result to the case with more than two blocks and provide a sufficient condition on the ensemble of $K$ matrices $\{A_i \}$ under which only $l \geq r/K $ random projections per block can capture the full row space. To do this, we define the following matrices:
\[M = \begin{bmatrix}
\phi A_1 \\
\phi A_2 \\
\vdots \\
\phi A_K \\
\end{bmatrix}  \ \text{and} \ \widehat{M} =  \begin{bmatrix}
V_{\calS} A_1 \\
V_{\calS} A_2 \\
\vdots \\
V_{\calS} A_K \\
\end{bmatrix} \label{eq:Mhat} \] where $V_{\calS}$ is any orthonormal matrix of size $l \times m$.

\begin{theorem}
\label{thm:manyK}
Given an ensemble of $K$ matrices $\{A_i\} $ for $i = 1,\cdots,K$, each of size $m \times n$ and a matrix $\phi \in \R^{l \times m}$ with entries drawn from the standard normal distribution, $M$ as defined above is full row rank if there exists an orthonormal basis $V_{\calS} \in \R^{l \times m}$ such that the $kl \times n$ size matrix $ \widehat{M}$ has full row rank. Consequently, for $l = r/K$, $\|(I - P_{M^T})A^T \| = 0$ with probability 1.
\end{theorem}

Intuitively, Theorem \ref{thm:manyK} requires that there is at least one subspace of dimension $l = r/K$, which when projected onto the matrices $A_i$ results in a set of $K$ linearly independent subspaces. In Figure \subref*{fig:hist}, we show the histogram of the smallest singular values of the matrix $\widehat{M}$ for the array imaging operator with $K= 8$ excitation wavelengths placed uniformly between $\lambda_{\min} = 7.5$cm and $\lambda_{\max} = 15$cm,  for $1000$ realizations of randomly chosen orthonormal basis $V_{\calS}$. In this case, the number of array elements was $M = 213$ and the scene considered had delta thickness. Hence, approximately only $30$ spatial measurements suffice in imaging any such scene. As the range extent of the images increases, the nested structure in the row spaces ceases to exist. However, the subspaces are still have a high degree of overlap and a number of aperture codes much smaller than the number of conventional beams used suffice for imaging. 

\begin{figure*}[tbhp]
\begin{minipage}{0.48\textwidth}
\subfloat[]
{\label{fig:eigU}\includegraphics[scale=0.2]{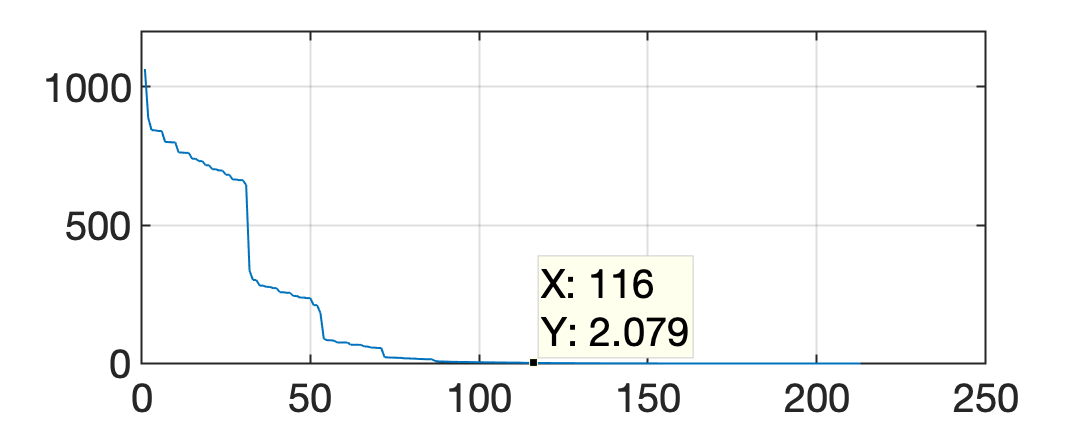}}
\end{minipage}
%\hspace*{\fill}
\quad
\begin{minipage}{0.48\textwidth}
\subfloat[]
{\label{fig:hist}\includegraphics[scale=0.2]{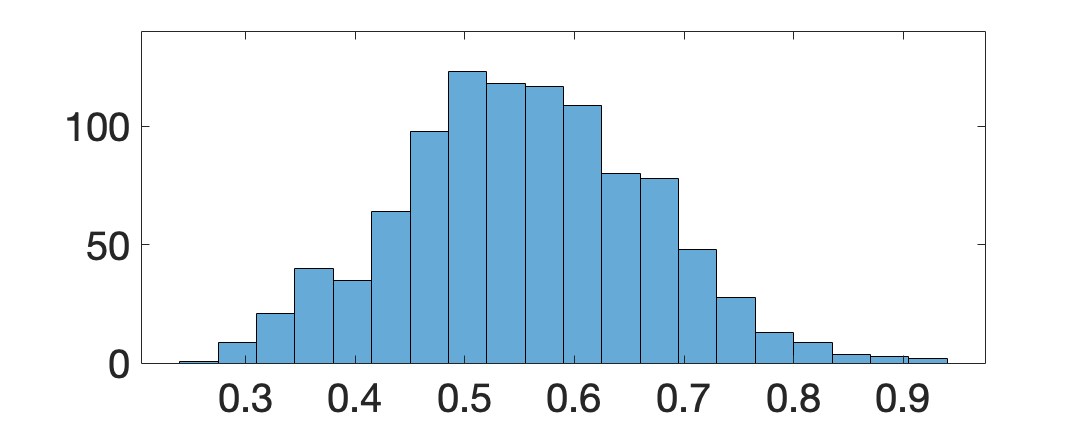}}
\end{minipage}
\caption{\small \sl (a) The eigenvalue distribution of $A_1A_2^T$ for two exctiation wavelengths $7.5$cm and $10$cm. The eigenvalue distribution ensures that $l = M/2$ random projections suffice (b) Histogram of the smallest singular value of $\widehat{M}$ over 1000 realizations of randomly generated orthobases $V_\mathcal{S}$, for the array imaging operator with $8$ excitation wavelengths were used between $7.5$cm and $10$cm. There exist many orthobases such that the sufficient condition of Theorem \ref{thm:manyK} holds. In this case, hence, $l = M/8$ random aperture codes are enough for imaging.}\label{fig:thm_conditions}
\end{figure*}

The conditions stated in Theorems \ref{thm:final_k2} and \ref{thm:manyK} ensure that the lack of diversity among the diagonal blocks of the RBD matrix is compensated for by the data matrix $A$ itself. As we show in further sections, the array imaging operator satisfies these conditions, thus lending itself to spatial subsampling.

\section{Extensions to aperture coded imaging}

The proposed aperture coded imaging system requires obtaining linear combinations of the array outputs at different excitation wavelengths. The analysis presented in the previous section assumed that the weights for these linear combinations are drawn from a continuous distribution such as the standard normal distribution and also assumed noise-free conditions. In this section, we wish to point out an acquisition method that follows a similar spirit but can be practically more appealing and we also model noisy imaging settings. Although we do not provide analysis for these cases, our simulations, shown in the next section, demonstrate that the fundamental idea of trading spatial measurements with excitation bandwidth for range limited images can still offer significant gains. 

\subsection{Imaging with subsampled array}
\label{subsec:subsampling}

The observation that the row spaces of the operator at different wavelengths overlap to a high degree for range-limited images suggests that we can also directly subsample the array and obtain direct measurments without any loss in resolution. Subsampled arrays can be thought of as aperture coded arrays with binary codes. Again drawing a parallel to the literature in randomized numerical linear algebra, subsampling the array is similar to  sampling a few rows and columns of a matrix to provide an approximation. (See \cite{pmlr-v32-ma14} and the references therein.) A standard approach in the linear algebraic community is to use randomized subsampling. A similar scheme can be applied to antenna arrays: sample only a fraction of randomly chosen array elements with broadband excitation. This results in a set of array measurements at different wavelengths that share the same subsampling pattern. 

There could also be more principled subsampling approaches for the specific application of broadband array imaging. Considering Figure \ref{fig:ppgridCR}, one might expect that sampling a few terminal elements and some elements at the center may be sufficient to obtain enough samples to reconstruct the common function. Providing theoretical justification for how many array elements are required could be an interesting direction for future research. In Figure \ref{fig:Mask}, we show the binary mask resulting from such an approach described for a 2D array. In the case of a scene with delta thickness present at a known constant range, measurements at each wavelength only samples the \textit{innovation} that the wavelength adds compared to other larger wavelengths. 

\begin{figure}[!tb]
\centering
\includegraphics[scale=0.3]{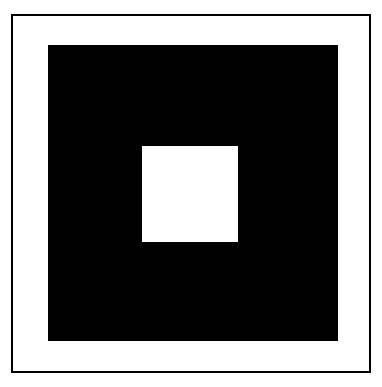}
\caption{\small \sl A binary mask that samples elements at the edge of the array and the center of the array. The white region represents the elements sampled. }
\label{fig:Mask}
\end{figure}

\subsection{Imaging in a noisy scenario}
\label{subsec:noisy}
The theorems from previous section address the case of noiseless imaging. However, we show in our experiments that aperture coding is also robust to noise. Stable least squares reconstruction is commonly achieved using Tikhonov regularization. Noisy measurements are modeled as  
\begin{equation}
y = Ax_0 + e
\label{noisy}
\end{equation} where $e$ represents the noise vector. $x$ can be estimated by solving 
\begin{equation}
x_{\text{Tikh}} = \mathop{\text{arg} \min}_x \|y - Ax\|_2^2 + \delta \| x\|^2
\label{tikh}
\end{equation} 
where $\delta$ is a regularization parameter. Robust reconstruction can be obtained using the right choice of the regularization parameter. For aperture coding, we assume the following model for noisy measurements, just as in traditional beamforming methods:
\begin{equation}
\Phi y = \Phi (Ax + e).
\label{noisy_compressed}
\end{equation}Again, we estimate $x$ by solving 
\begin{equation}
x_{C, \ \operatorname{Tikh}} = \text{arg} \min \|\Phi y - \Phi Ax\|_2^2 + \hat{\delta} \| x\|^2.
\end{equation}
The choice of $\delta$ and $\hat{\delta}$ affect the reconstruction performance. While increasing the regularization parameter reduces the effect of noise on the reconstruction, it also leads to an increase in the signal reconstruction performance. The choice of the parameter hence leads to a trade-off. In our simulations, we show that for a given choice of $\delta$, we can choose $\hat{\delta}$ to match the signal reconstruction performance and the noise performance of traditional imaging. 

We use the following definitions to quantify the performance of aperture coding in a noisy scenario:
\begin{equation}
 \text{Signal reconstruction error}~=~\frac{\|x_{\text{LS}} - (\Phi A)^{\dagger}_{\operatorname{Tikh}} \Phi Ax \|^2}{\| x_{\text{LS}}\|^2} 
\end{equation}
\begin{equation}
\text{Output SNR}~=~20 \log_{10} \frac{\|(\Phi A)^{\dagger}_{\operatorname{Tikh}} \Phi Ax  \|}{\|(\Phi A)^{\dagger}_{\operatorname{Tikh}} \Phi e\|}
\label{eq:opSNR}
\end{equation}
where $(S)^{\dagger}_{\operatorname{Tikh}} = (S^TS + \delta I)^{-1}S^T$ denotes the Tikhonov regularized pseudo-inverse of the matrix $S$.
\section{Experiments}
\label{sec:exp}

We now provide  experimental results to show the effectiveness of aperture coding. Various experiments were conducted: aerture coding simulations were conducted with images at a constant depth from the antenna array, for flat images parallel to the 2D array and delta thickness images that are multi-depth with known and unknown range profiles. The next set of experiments deal with subsampling the array for the same class of images. These are followed by simulations in the presence of noise.  We then provide simulation results for objects whose range limits are not thin and span a range of $15$cm. The following array parameters were used in all the simulations: an array of $40\times 40$ elements were used with 15 excitation wavelengths placed regularly in the bandwidth of $7.5$cm and $15$cm. The elements were placed at half the smallest wavelength. The scene was assumed to be within the angular span of $[-\pi/4,\  \pi/4]$ in both directions. For this configuration the number of beams/ measurements needed at each frequency for conventional imaging is about 1100. Quantitative error values for all the experiments are given in Table \ref{table:errors}. 

\textbf{A. Constant range and multi-depth images}

Figure \ref{fig:CRmain} shows reconstruction results for images at a constant, known range. Conventional beamforming requires about 1100 beams in this case to scan over the entire image. It is clear that similar reconstruction performance can be obtained using as few as 80 beams with wideband excitation. The relative reconstruction error is almost negligible in each case (shown in Table \ref{table:errors}).

\begin{figure}[tbhp]
\centering
\includegraphics[scale=0.25]{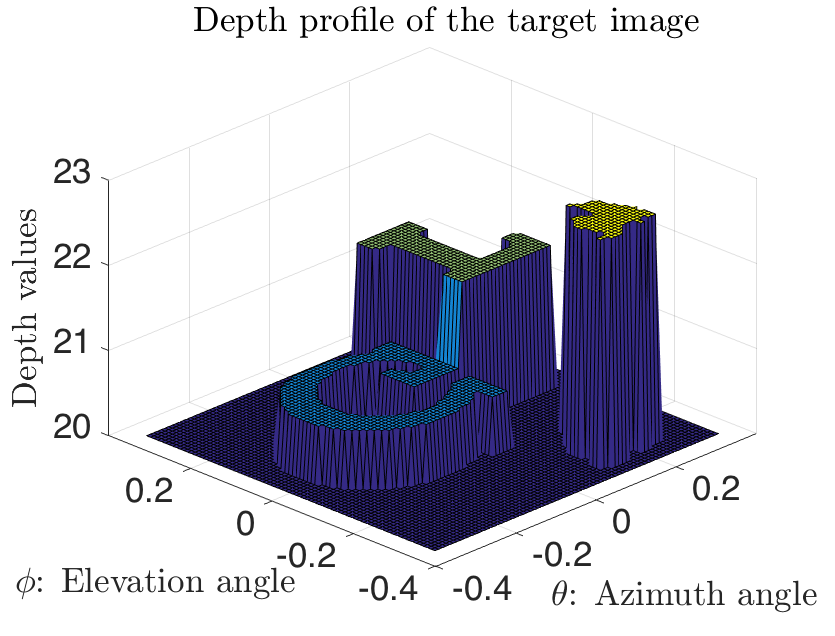}
\caption{\small \sl Depth map of the multi-depth image used in simulations.}
\label{DM}
\end{figure}

\begin{figure}[tbhp]
\centering
\begin{minipage}{0.32\linewidth}
%\centering
\subfloat[\small \sl full imaging]{\label{MD:a}\includegraphics[scale=.22]{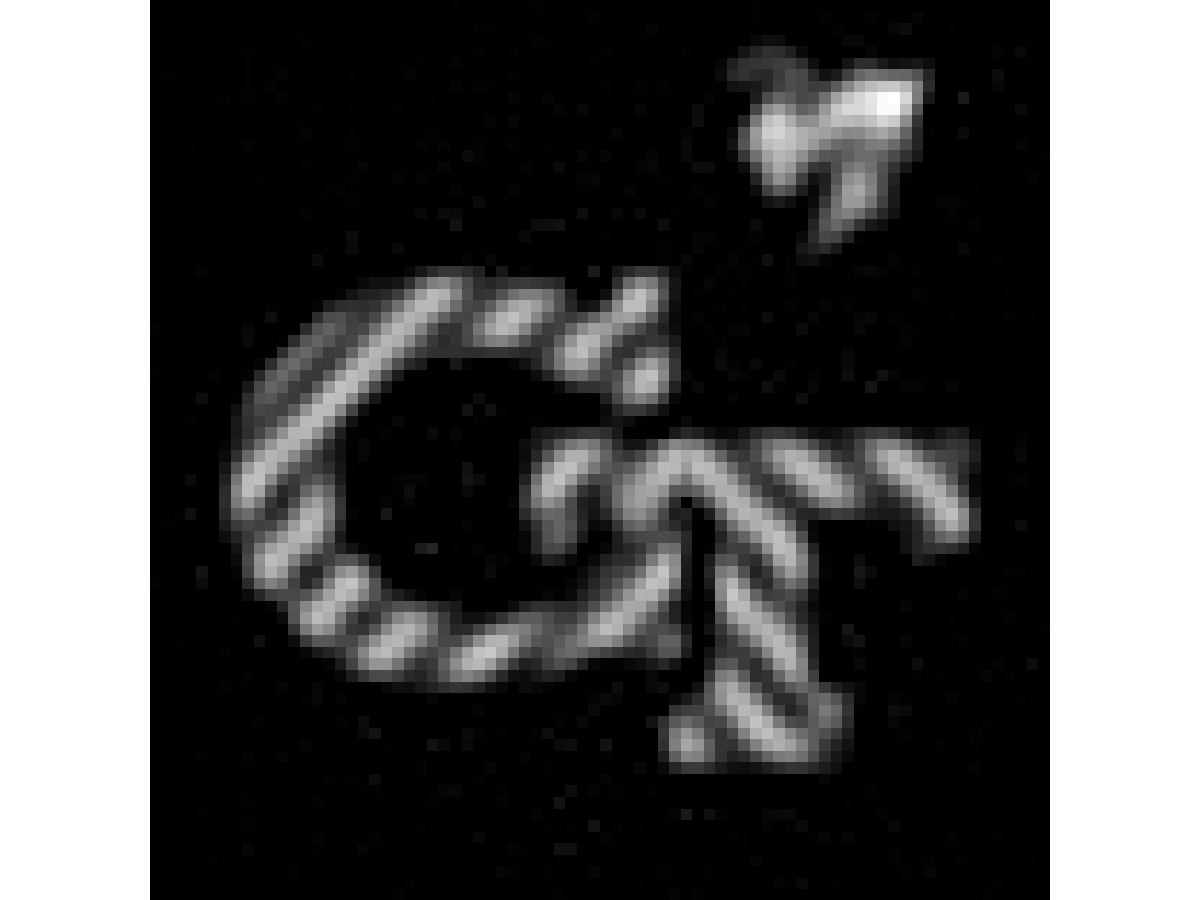}}
\end{minipage}
\hspace{1.5em}
\begin{minipage}{0.32\linewidth}
%\centering
\subfloat[\small \sl 320 generic beams]{\label{MD:c}\includegraphics[scale=.22]{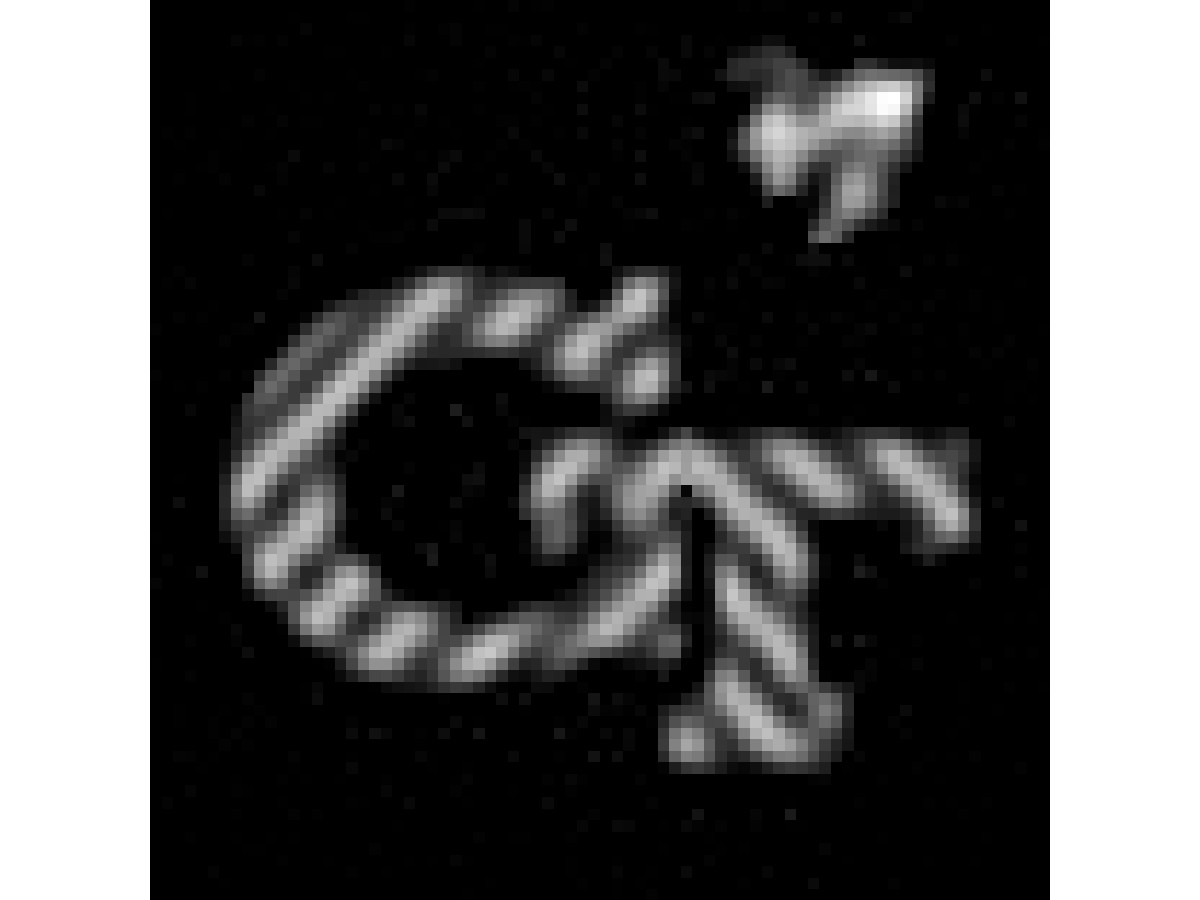}}
\end{minipage}
\par\medskip
\begin{minipage}{0.32\linewidth}
%\centering
\subfloat[\small \sl 160 generic beams]{\label{MD:d}\includegraphics[scale=.22]{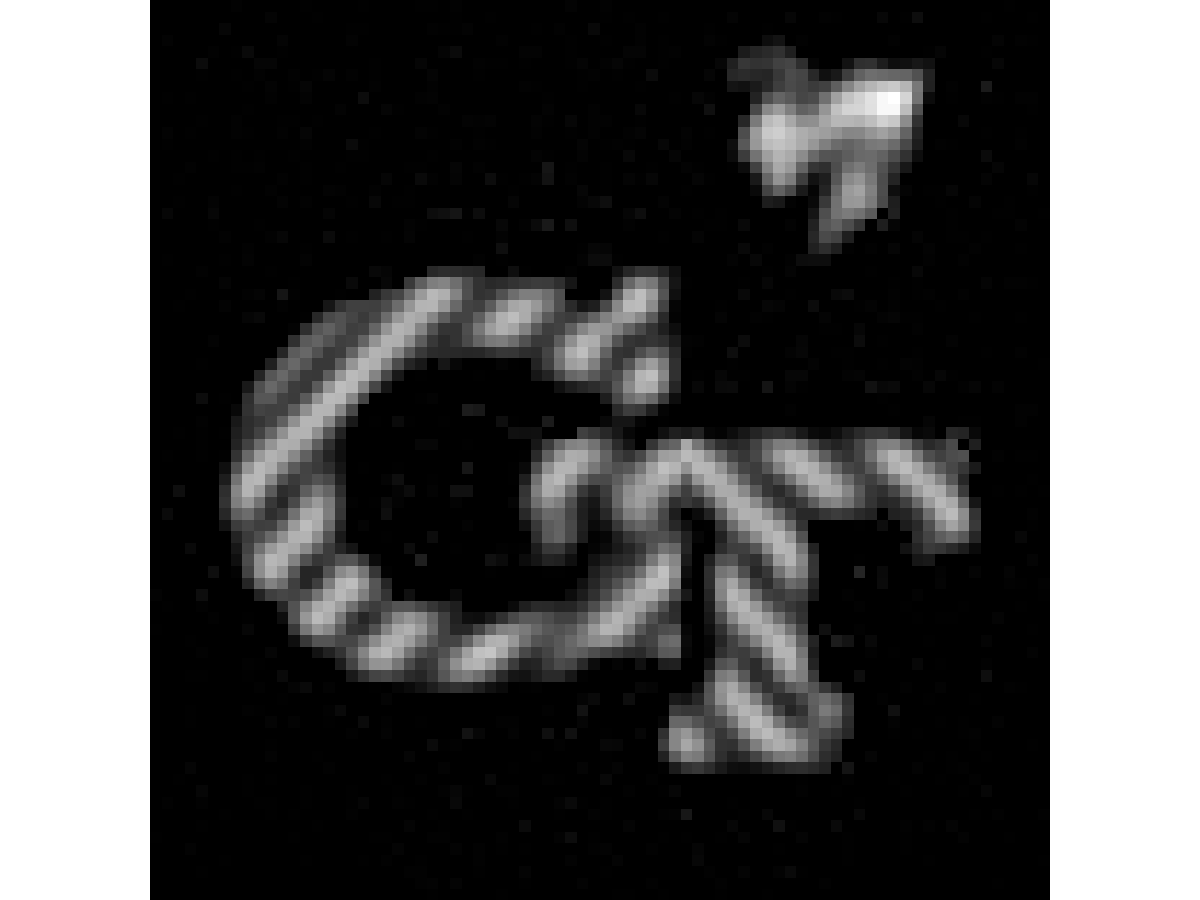}}
\end{minipage}
\hspace{1.5em}
\begin{minipage}{0.32\linewidth}
%\centering
\subfloat[\small \sl 80 generic beams]{\label{MD:e}\includegraphics[scale=.22]{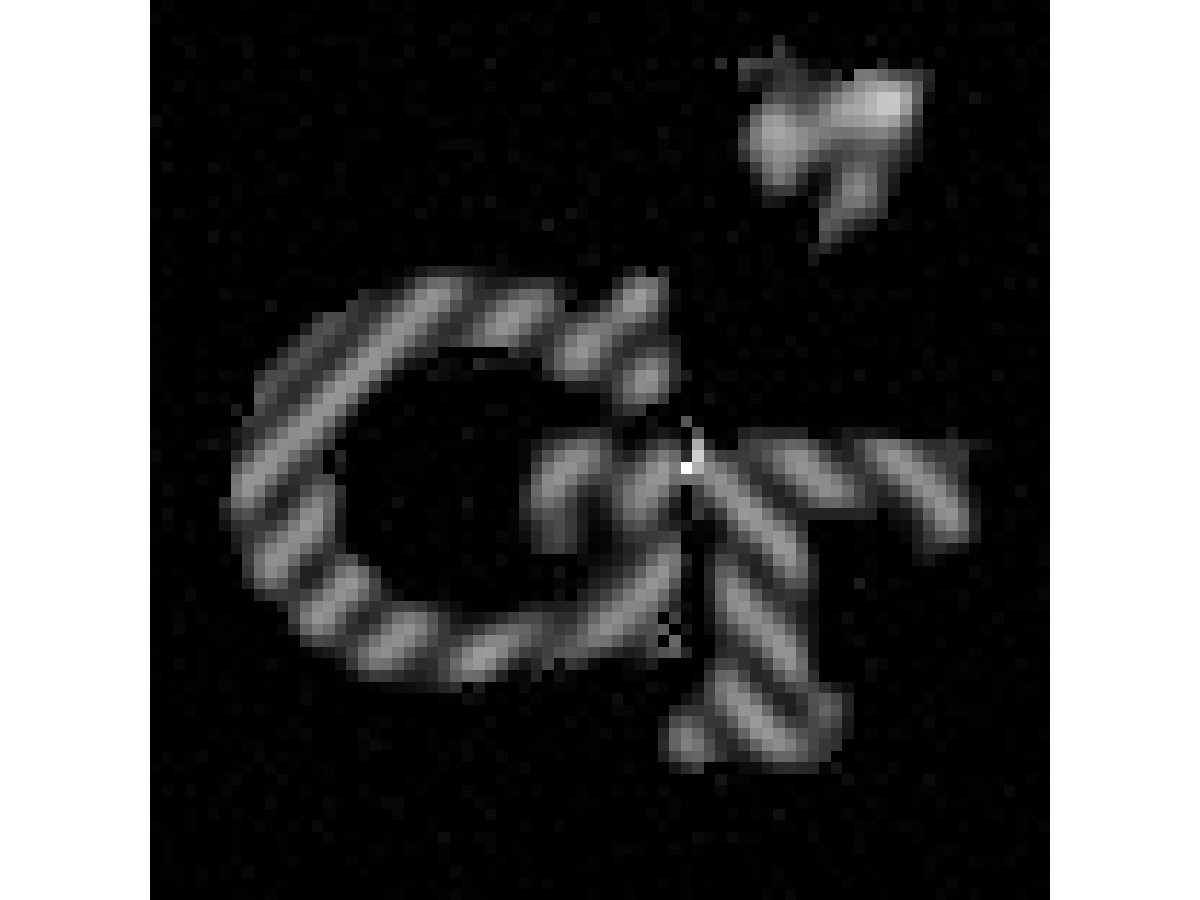}}
\end{minipage}

\caption{\small \sl Aperture coded imaging for a multi-depth image. (a) represents the conventional method, which uses about ~1100 beams. (b), (c) and (d) show reconstruction results with $320$, $160$ and $80$ generic beams.}
\label{fig:MDmain}
\end{figure}

Aperture coding is also effective for multi-depth images. A scene with three segments, each at a different depth was considered. The depth map of the scene used in the simulations is shown in Figure \ref{DM}. In Figure \ref{fig:MDmain}, we present simulation results when the depth profile was assumed to be known a priori. Again only about 80 beams are sufficient to get good quality reconstructions. 

We further explore the performance of coded aperture imaging when the depth profile is unknown. Boufounos in \cite{6289146} describes a method to infer the depth profile of an image using CoSaMP algorithm. The algorithm uses full array measurements in the least squares step. We replace this step with sketched least squares and coded measurements. Simulation results are shown in Figure \ref{fig:MDUDmain}. We note that aperture coding still performs well, with the relative error being negligible even for 160 generic beams.

\begin{figure}[!tb]
\centering
\begin{minipage}{.32\linewidth}
\centering
\subfloat[\small \sl full imaging]{\label{MDUD:a}\includegraphics[scale=.23]{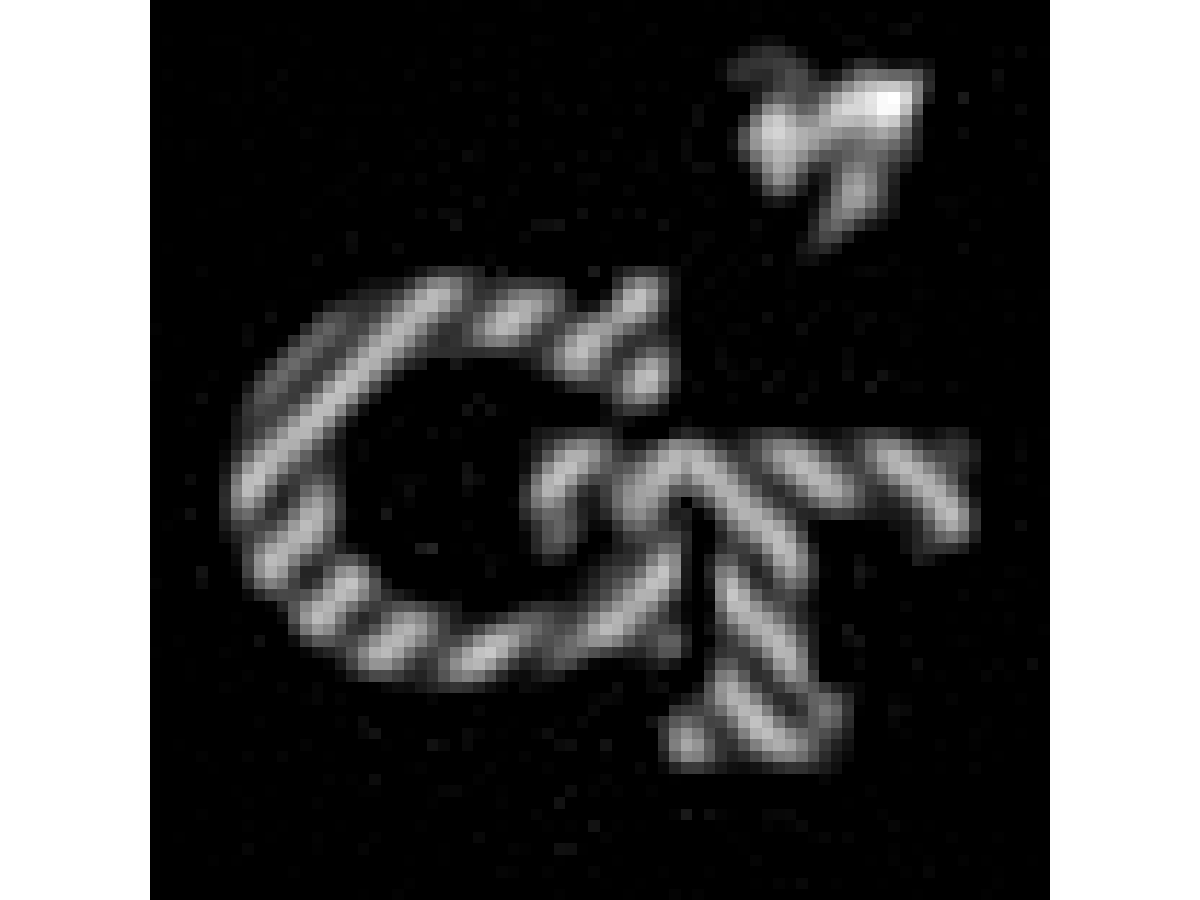}}
\end{minipage}%
\hspace{1.4em}
\begin{minipage}{.32\linewidth}
\centering
\subfloat[\small \sl 320 generic beams]{\label{MDUD:c}\includegraphics[scale=.23]{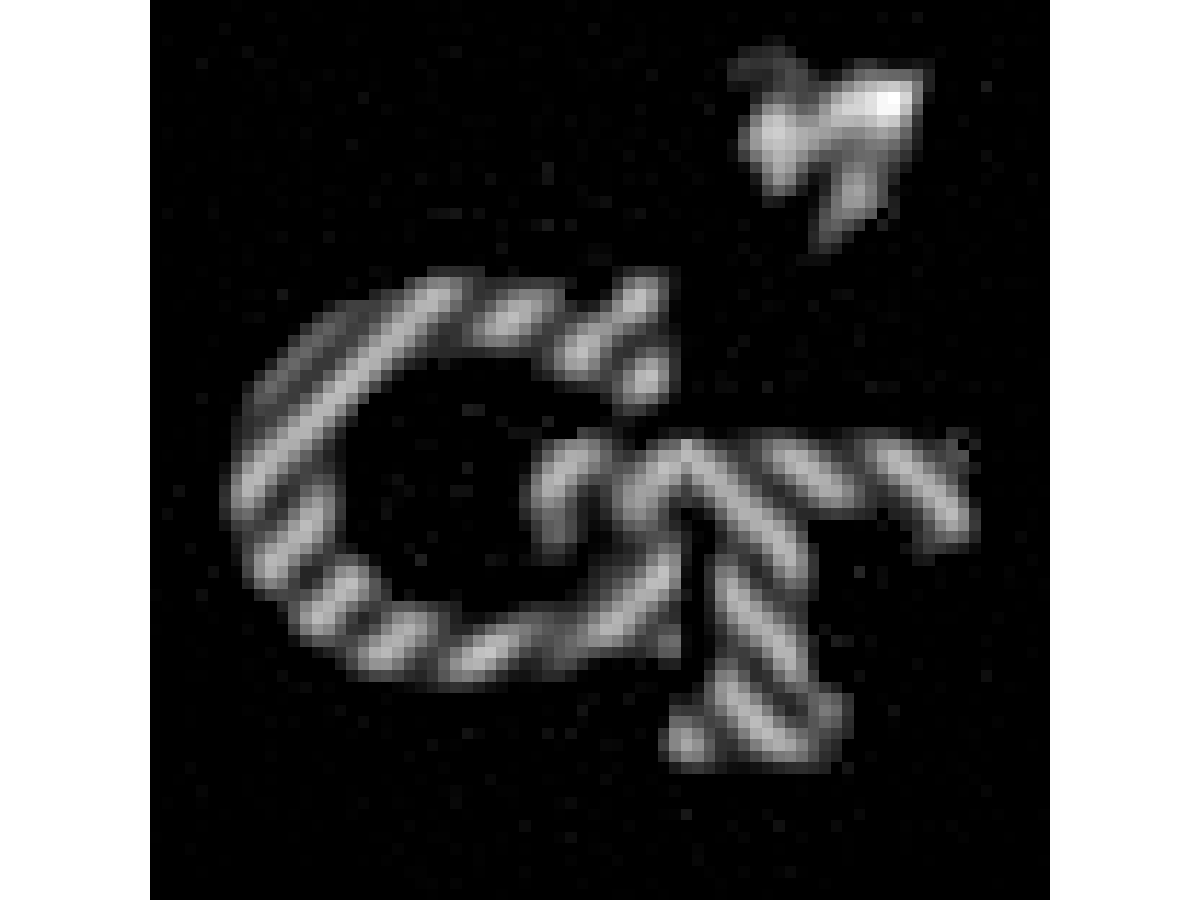}}
\end{minipage}
\par\medskip
\begin{minipage}{.32\linewidth}
\centering
\subfloat[\small \sl 160 generic beams]{\label{MDUD:d}\includegraphics[scale=.23]{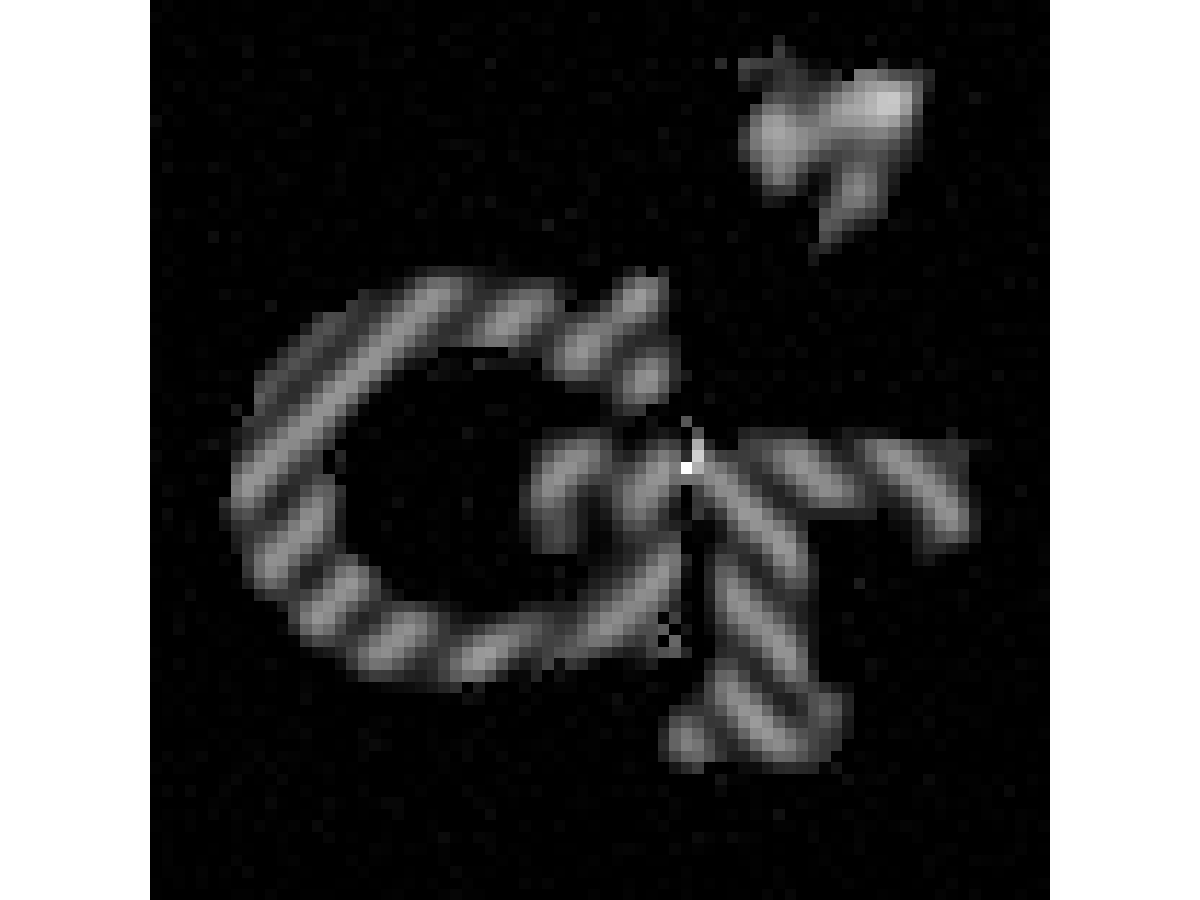}}
\end{minipage}
\hspace{1.4em}
\begin{minipage}{.32\linewidth}
\centering
\subfloat[\small \sl 80 generic beams]{\label{MDUD:e}\includegraphics[scale=.23]{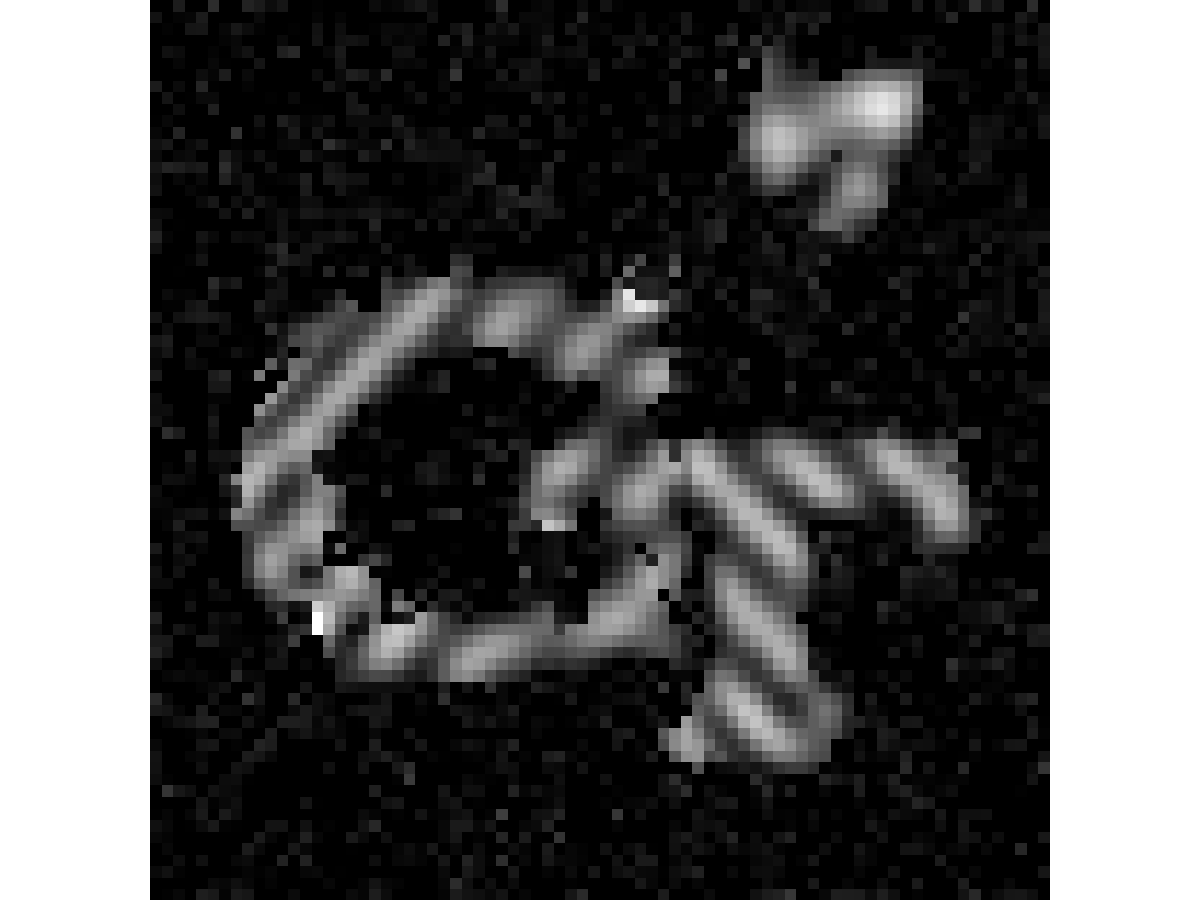}}
\end{minipage}

\caption{\small \sl Aperture coded imaging for a multi-depth image with an unknown range profile. (a) represents the conventional method, which uses about ~1100 beams.}
\label{fig:MDUDmain}

\end{figure}

\begin{table*}
\centering
\renewcommand{\arraystretch}{1.5}
\begin{tabular}{|c|c|c|c|c|c|c|}
\hline
\multirow{2}{*}[-5pt]{Imaging mode}& \multirow{2}{1.5cm}[-5pt]{\centering Constant range (CR)} & \multirow{2}{1.5cm}[-5pt]{\centering Flat surface (FS)} & \multirow{2}{2cm}{\centering Multi-depth, known \\ range profile (MD)} & \multirow{2}{2cm}{\centering Multi-depth, unknown range profile (MDUR)} & \multicolumn{2}{c|}{Constant range noisy} \\[6pt]
\cline{6-7}
& & & & &  Signal error & OP SNR \\[6pt]
\hline
Full (1100 beams) & NA & NA & NA & NA  & NA & 16.09 \\ %[2ex]
\hline

320 codes & 2.7e-5& 4.4e-4& 3e-4& 1.1e-2  & 2.5e-5 & 16.0117 \\%[2ex]
\hline

160 codes  & 7.4e-5 & 3.3e-4& 1.2e-3& 3.1e-2 & 6.3e-5 & 15.9112 \\%[2ex]
\hline

80 codes & 4.2e-4 & 5.1e-4& 4.8e-2& 1.6e-1 & 4.9e-4 & 15.3206 \\ %[2ex]
\hline
\end{tabular}
\caption{\small \sl Relative reconstruction error values for different classes of images. Aperture codes with weights chosen from standard normal distribution were used. The last column shows the output SNR as calculated using equation \eqref{eq:opSNR}.  }
\label{table:errors}
\end{table*}

\textbf{C. Imaging with noise} 

As claimed before, aperture coded imaging can still perform well in the presence of noise. We simulated an imaging scenario at $20$dB input SNR level. In Figure \ref{fig:noisy}, we show reconstruction results of these simulations. The corresponding signal reconstruction error values and output SNR are given in Table \ref{table:errors}. It is clear that aperture coded measurements do not result in any degradation in the presence of noise. This shows that the sketched matrix is not only full rank, but also has stable eigenvalues.

\begin{table}
\centering
\renewcommand{\arraystretch}{1.5}
\begin{tabular}{|c|c|c|c|c|}
\hline
Imaging mode & CR & FS & MD & MDUR \\
\hline
320 elements & 2e-4 & 1e-4 & 1e-3 & 3.8e-2 \\
\hline
160 elements & 9e-4 & 4e-4 & 3.1e-2 & 6.8e-2\\
\hline
80 elements & 1.6e-3 & 1.1e-3 & 5.8e-3 & 18.4e-2\\
\hline
\end{tabular}
\caption{\small \sl Relative reconstruction errors with subsampled array: The array was randomly subsampled to have 320, 160 and 80 elements. Abbreviations are as in Table \ref{table:errors}.}
\label{table:sub}
\end{table}

\textbf{D. Imaging scenes with higher range limits}

 Aperture coding can be highly effective even when imaging scenes with higher range limits. As predicted by our theorems, when the range limit increases, more measurements are required, but it can still less be than that used in conventional imaging. We demonstrate this in our next set of simulations. We consider an object made up of five layers as shown in Figure \ref{RL:a}. These layers lie in a region of width $15$cm, in the far-field of a 2D array at a depth of $20$m. We again use 15 excitation wavelengths between $7.5$cm and $15$cm. In Figure \ref{fig:RL}, we present reconstruction results by using 640, 480, and 320 aperture codes in place of 1100 beams used in conventional imaging.

\begin{figure}[!tb]
\centering
\begin{minipage}{.35\linewidth}
\centering
\subfloat[\small \sl full imaging]{\label{CRnoisy:a}\includegraphics[scale=.2]{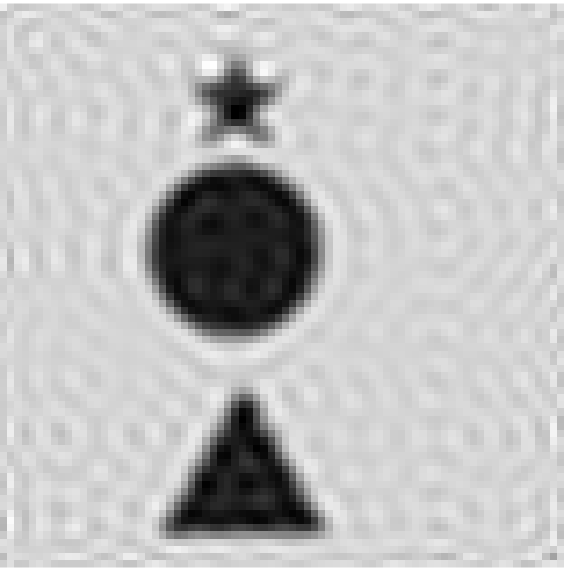}}
\end{minipage}%
\begin{minipage}{.35\linewidth}
\centering
\subfloat[\small \sl 320 generic beams]{\label{CRnoisy:c}\includegraphics[scale=.2]{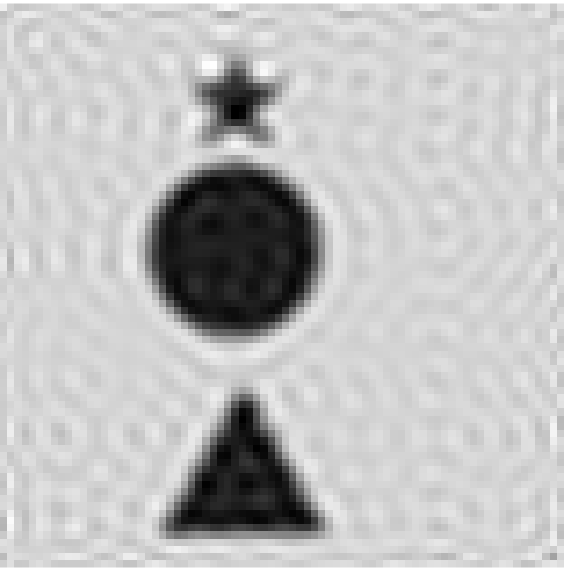}}
\end{minipage}
\par\medskip
\begin{minipage}{.35\linewidth}
\centering
\subfloat[\small \sl 160 generic beams]{\label{CRnoisy:d}\includegraphics[scale=.2]{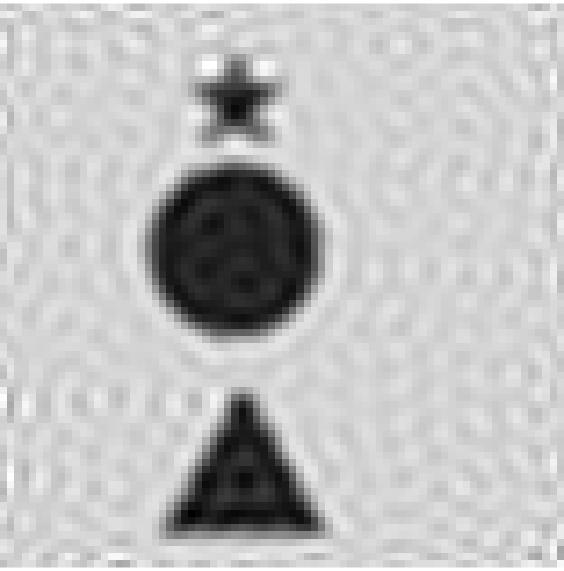}}
\end{minipage}
%\hspace{1em}
\begin{minipage}{.35\linewidth}
\centering
\subfloat[\small \sl 80 generic beams]{\label{CRnoisy:e}\includegraphics[scale=.2]{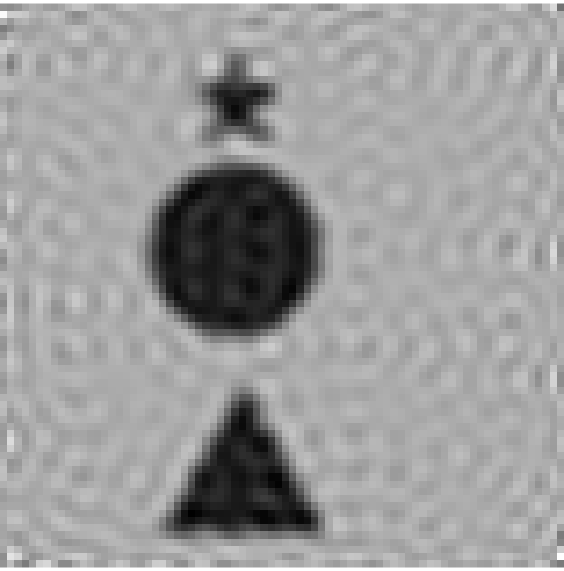}}
\end{minipage}

\caption{\small \sl Aperture coding in the presence of noise. The regularization parameters was varied in each case to match the SNR of the full imaging scenario. The noise performance was preserved without compromising on the signal reconstruction quality, as seen in Table \ref{table:sub}. }
\label{fig:noisy}

\end{figure}

\begin{figure*}[!tb]
\begin{minipage}{.5\linewidth}
\centering
\subfloat[\small \sl Full reconstruction of a set of slices of a sphere $\sim$ 1100 measurements]{\label{RL:a}\includegraphics[scale=0.28]{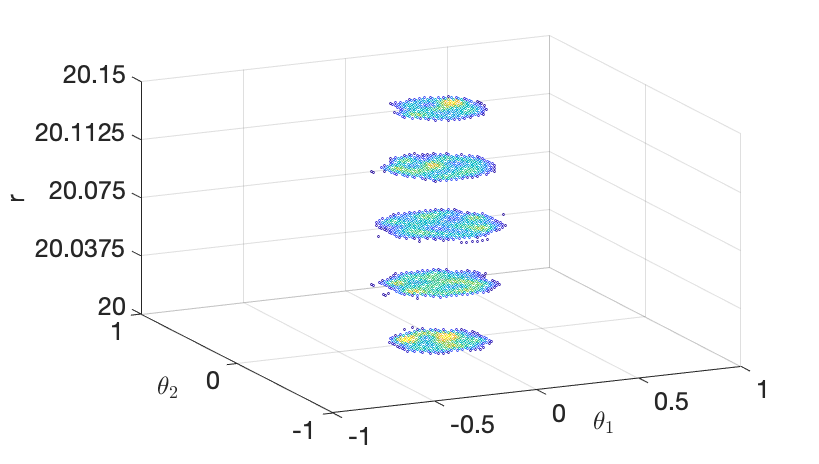}}
\end{minipage}
%\hspace{8em}
\begin{minipage}{.5\linewidth}
\centering
\subfloat[\small \sl 640 generic  beams, relative  error: 0.0168]{\label{RL:b}\includegraphics[scale=.28]{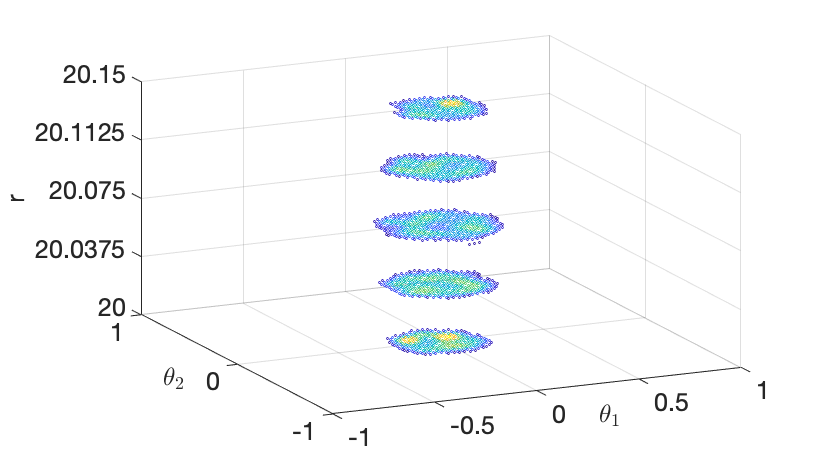}}
\end{minipage}
\par\medskip
\begin{minipage}{.5\linewidth}
\centering
\subfloat[\small \sl 480 generic beams, relative  error: 0.0298]{\label{RL:c}\includegraphics[scale=.28]{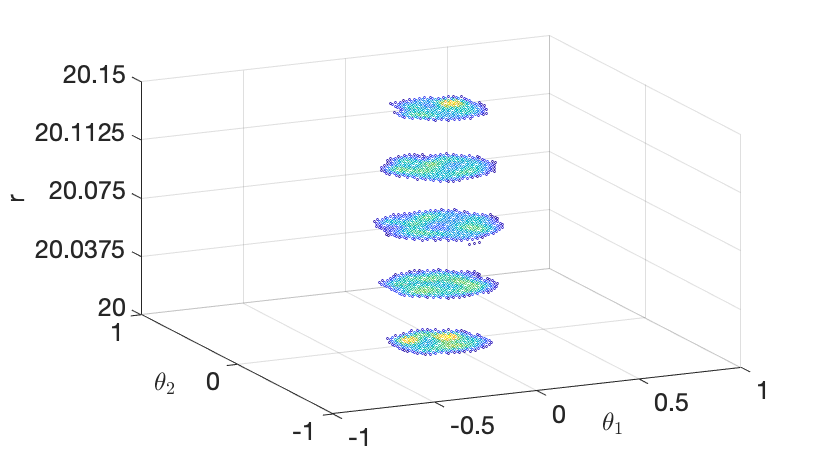}}
\end{minipage}
%\hspace{8em}
\begin{minipage}{.5\linewidth}
\centering
\subfloat[\small \sl 320 generic  beams, relative  error: 0.0466]{\label{RL:d}\includegraphics[scale=.28]{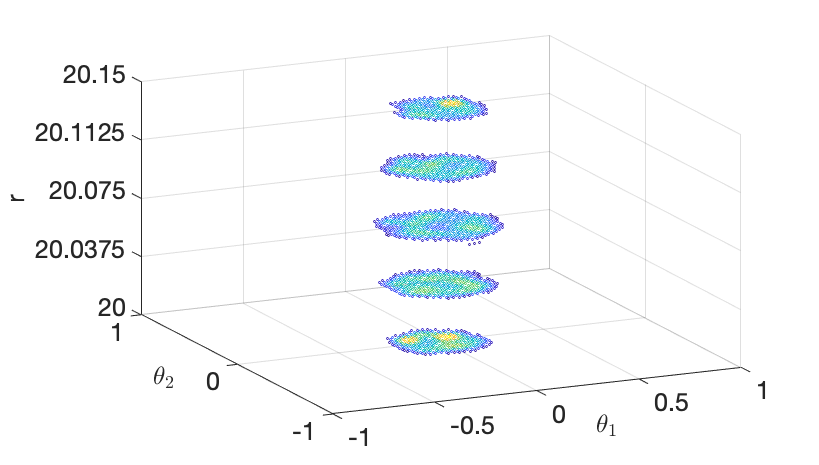}}
\end{minipage}
\caption{\small \sl Reconstruction results for a target object with (at $20$m) with aperture coded acquisition. The target scene consists of five discs one behind the other. In the figure above, the 2D array would face the discs from below. Using a set of $15$ excitation wavelengths between $7.5$cm and $15$cm such a scene can be reconstructed using only about 320 spatial measurements, unlike conventional beamforming which would require about $1100$ measurements.}
\label{fig:RL}
\end{figure*}

%\balance

\section{Conclusion}

In this paper, we identify that broadband array imaging of range-limited target scenes is a particular case of spatio-spectral concentration, which results in the reconstructed image having a limited number of degrees of freedom. However, any image acquisition method that takes advantage of this low dimensional structure has to also be practically feasible. We show both theoretically and using simulations that existing array architectures such as aperture codes and subsampled arrays offer a way to make use of the low dimensionality to perform array imaging with far fewer spatial measurements than conventional methods when broadband excitation is used. 

To establish the theoretical guarantees, we modeled the array imaging problem as a sketched least squares problem with a particular sketching matrix that has a block diagonal structure. We show that  even with such structured randomness, the sketching matrix can be as small/compressive as a more generic dense sketching matrix. In terms of array imaging, this implies that existing array architectures can be used along with broadband excitation to image range-limited target scenes with very few spatial measurements.

\appendix

\section{ Proof of Theorem \ref{thm:final_k2}}
The main tool we use to prove Theorems \ref{thm:final_k2} and \ref{thm:manyK} is polynomial identity testing: if a polynomial of a finite total degree is said to be identically zero if the coefficient of every monomial term is zero. Any polynomial which is not identically zero, when evaluated at a random point drawn from a continuous distribution, is non-zero with probability 1, since the set of roots of the polynomial has measure zero with respect to the field of real numbers.  In conclusion, a multivariate polynomial evaluated at a multi-dimensional point drawn from a continuous distribution is non-zero with probability 1, if any of the polynomial coefficients are non-zero. 

To show that a given matrix is of rank $r$, we choose a suitable submatrix of size $r \times r$ and show that it has a determinant with at least one non-zero coefficient and therefore not identically zero. For random projections whose entries are drawn from the standard normal distribution, we hence need to show that the sketched matrices have submatrices of desirable sizes that have a non-zero determinant with probability 1.

For ease of notation, we prove the result on the transposed matrices: $A = [A_1 \ A_2]$, $\|(I - P_M)A \| = 0$, $M = [A_1 P \  A_2 P]$, and $P$ is the sketching matrix with entries drawn from the standard normal distribution. For ease of explanation, we assume that $A_1$ and $A_2$ are orthonormal matrices: $A_i^TA_i = I$ and $\calR(A_1) = \calR(A_2)$. We later explain how our proof holds for any two general matrices. 

Consider the matrix  $Z = A_1^TM = [P \ \ \ UP]$ where $U = A_1^TA_2$ is an orthogonal matrix. If $Z$ is full column rank, then $\operatorname{rank}(M) $ is full column rank. 
Any orthogonal matrix can be decomposed as 
\begin{equation}
\label{eq:canonical_ortho}
U = Q^TRQ
\end{equation}
where $Q$ is an orthobasis and $R$  is a block diagonal matrix with $2 \times 2$ or $1 \times 1$ blocks. The $2 \times 2$ blocks are of the form 
\[ R_i = \begin{bmatrix}
\cos\theta_i & -\sin\theta_i \\
\sin\theta_i & \cos\theta_i
\end{bmatrix} \]
with $\cos\theta_i\pm j\sin\theta_i$ being a pair of complex conjugate eigenvalues of $U$, $ j = \sqrt{-1}$. The $1\times 1$ diagonal blocks are equal to $\pm 1$ and are also eigenvalues of $U$. \eqref{eq:canonical_ortho} is referred to as the canonical decomposition of $U$. 

Notice that $\text{rank}([P \ \ \ UP]) = \text{rank}(Q [P \ \ \ UP]) = \text{rank}([QP \ \ \ RQP])$. Since $Q$ is orthogonal, $QP$ is also a matrix with i.i.d.\ standard normal variables due to the rotational invariance of the standard normal distribution. Thus we can directly work with the matrix $Z = [P \ \ \ RP]$. The following lemma provides the necessary condition on the multiplicity of any real eigenvalue for $[P \ \ \ RP]$ to be full column rank.

\begin{lemma}
Let there be a real eigenvalue of $U$ that has an algebraic multiplicity $n_1 > n/2$. Then for some $n - n_1 < l <n/2$, $Z = [P \ \ \ RP]$ is not full column rank. 
\end{lemma}

\begin{proof}
Let $\lambda_o$ be a real eigenvalue with algebraic multiplicity  $n_1 > n/2$ and let $n_2 = n - n_1$. Then 
\begin{equation}
\label{eq: necessity_condition}
Z = [P \ \ \ RP] = \begin{bmatrix}
P_1 & \lambda_o P_1 \\
P_2 & R_2P_2
\end{bmatrix}
\end{equation} where $P_1$ is the submatrix formed by the first $n_1$ rows of $P$ and $R_2$ is the submatrix formed by the last $n_2$ rows and columns of $R$. Let $l = n_2/2 + q $, $q > n_2/2$. Then, the submatrix $[P_2 \ \ \ R_2P_2]$ has a null space $\mathcal{N}_1$ of dimension $2q > n_2$ and the submatrix $[P_1 \ \ \ \lambda_o P_1]$ has an $l$ dimensional null space $\mathcal{N}_2$ given by the range of $\begin{bmatrix}
-\lambda_o I \\
I
\end{bmatrix}$. Since $2q + l > 2l$, $\mathcal{N}_1 \cap \mathcal{N}_2 \neq \phi$. The result follows. 
\end{proof}

We now prove some lemmas that provide sufficient conditions on $R$ for which $[P \ \ \ RP]$ is full column rank. We state the lemmas along with their proofs and then use them to prove Theorem \ref{thm:final_k2}. 

\begin{lemma}
\label{thm:conjugate_eig}
Let there be $l_o$ pairs of complex conjugate eigenvalues of $R$ with non-zero imaginary part. Then for any $l \leq l_o$, $Z = [P \ \ \ RP]$ has full column rank.
\end{lemma}

\begin{proof}
 Let $P$ be expressed as 

\begin{equation*}
P = \begin{bmatrix}
P_1 \\
P_2\\
\vdots \\
P_{\frac{n}{2}}
\end{bmatrix}
\end{equation*} where $P_i \in \R^{2 \times l}$. Let $P_{i,j,k}$ denote the $(j,k)$th element of $P_i$. Rearranging the columns of $Z$, we obtain

\begin{equation}
\label{eq:z_rearranged}
\widehat{Z}_1 = \begin{bmatrix}
p_{1,1,1} & \cos\theta_1 p_{1,1,1} - \sin\theta_1 p_{1,2,1} & \cdots  \\
p_{1,2,1} & \sin\theta_1 p_{1,1,1} - \cos\theta_1 p_{1,2,1} & \cdots \\
\vdots & \vdots &  \\
p_{\frac{n}{2},1,1} & \cos\theta_1 p_{\frac{n}{2},1,1} - \sin\theta_1 p_{\frac{n}{2},2,1} & \cdots  \\
p_{\frac{n}{2},2,1} & \sin\theta_1 p_{\frac{n}{2},1,1} - \cos\theta_1 p_{\frac{n}{2},2,1} & \cdots
\end{bmatrix}
\end{equation}

Expanding the determinant of $\widehat{Z}$, the coefficient of the term $\underset{i=1,2,\cdots,l} {\prod}p_{i,1,i}^2$ is $\underset{i = 1,2,\cdots,l}{\prod} \sin\theta_i$. For any $l \leq n/2$, $\det(\widehat{Z}(2lj-2j+1:2lj, 2lj-2j+1:2lj))$ has a term of the form $\underset{i = 2j-1...2j+l-1}{\prod}p_{j,1,j}^2$ with coefficient $\underset{i = 2j-1...2j+l-1}{\prod}\sin\theta_i$. If there is a set of $l$ rotations $R_i$ such that $\theta_i \neq 0 \ \forall \ i$, then clearly, the determinant associated with this $2l \times 2l$ diagonal block in $\widehat{Z}$ is not equal to $0$ with probability $1$ and the result follows.
\end{proof}

Let the number of complex conjugate pairs of eigenvalues be $n_1$. We next group as many real eigenvalues as possible into pairs such each pair has distinct eigenvalues. Let the number of such eigenvalue pairs be $n_2$. The remaining eigenvalues are all real and equal to each other. Let the number of such eigenvalues be $2n_3$ (this number will be even since $n$ is even).

\begin{lemma}
\label{thm:pairs}
For $l \leq n_2$, $Z = [P \ \ \ RP]$ has full column rank.
\end{lemma}

\begin{proof}
Again rearranging the columns of $Z$ as 

\begin{equation}
\label{eq:distinct_pairs}
\widehat{Z}_2 = \begin{bmatrix}
P_{1,1,1} & \lambda_{1,1} P_{1,1,1} & \cdots &  P_{1,1,l} & \lambda_{1,1} P_{1,1,n_2} \\
P_{1,2,1} & \lambda_{1,2} P_{1,2,1} & \cdots &  P_{1,2,l} & \lambda_{1,2} P_{1,2,n_2} \\
\vdots & \vdots &  & \vdots & \vdots \\
P_{\frac{n}{2},2,1} & \lambda_{n_2,2} P_{\frac{n}{2},2,1} & \cdots &  P_{\frac{n}{2},2,n_2} & \lambda_{n_2,2} P_{\frac{n}{2},2,n_2} \\
\end{bmatrix}
\end{equation} 
Here $\lambda_{i,j}$ denotes the $j^{\text{th}}$ eigenvalue of the $i^{\text{th}}$ pair. The coefficient of the $2l$ the degree term $\underset{i = 1,..,l}{\prod}{P_{i,1,i}P_{i,2,i}}$ is $\underset{i = 1,..,l}{\prod}{\lambda_{i2} - \lambda_{i1}} \neq 0$.  The determinant associated with this $2l \times 2l$ diagonal block in $\widehat{Z}$ is not equal to $0$ with probability $1$. The result follows.
\end{proof}

\begin{lemma}
\label{thm:repeated}
If $n_3 < n_1$, then for $l < 2n_3$, $Z = [P \ \ \ RP]$ has full column rank with probability 1. 
\end{lemma}

\begin{proof}
Consider the submatrix of $Z$ shown below, after rearranging the rows and columns

\begin{equation}
\label{eq:n3n1}
\widehat{Z}_3 = \begin{bmatrix}
P_{1,1,1} & c\theta_1 P_{1,1,1} - s\theta_1 P_{1,2,1} & \cdots \\
P_{n_1 + n_2 + 1,1,1} & \lambda_o P_{n_1 + n_2 + 1,1,1}  & \cdots  \\
P_{1,2,1} & s\theta_1 P_{1,1,1} + c\theta_1 P_{1,2,1} & \cdots  \\
P_{n_1 + n_2 + 1,2,1} & \lambda_o P_{n_1 + n_2 + 1,2,1}  & \cdots \\
\vdots & \vdots &  \\
P_{n_3,2,1} & s\theta_{n_3} P_{n_3,1,1} +c\theta_{n_3} P_{n_3,2,1} & \cdots \\
P_{n_1 + n_2 + n_3,2,1} & \lambda_o P_{n_1 + n_2 + n_3,2,1}  & \cdots
\end{bmatrix}
\end{equation}

The coefficient of the $2l^{\text{th}}$ degree term $\underset{i = 1,..,l/2}{\prod} \underset{k = 1,2}{\prod} P_{n_1+n_2+i,k,(i-1)*2+k}P_{i,2,(i-1)*2+k} $ in the determinant of $\widehat{Z}_3$ is equal to $ \underset{i=1,..,l/2}{\prod} \sin^2\theta_i$.

\end{proof}

\begin{lemma}
\label{lemma_n3n1}
If $n_3 > n_1$, then there is a real eigenvalue $\lambda_o$ with algebraic multiplicity greater than $n/2$.
\end{lemma}

\textbf{Proof of Theorem \ref{thm:final_k2}:}
With eigenvalues grouped as before, let $\theta_i$ parametrize the $2 \times 2$ diagonal block of $R$ due to the $i^{\text{th}}$ pair of complex conjugate eigenvalues, $\lambda_{j1}, \lambda_{j2}$ denote the $2 \times 2$ diagonal block due to the $j^{\text{th}}$ pair of real eigenvalues with $\lambda_{j1} \neq \lambda_{j2}$, and let the rest of eigenvalues be equal $\lambda_o$ (repeated at least $n_3$ times). Due to Lemma \ref{lemma_n3n1}, $n_3 < n_1$. Let $n_4 = n_1 - n_3$, and let $l = n/2 = n_3 + n_2 + n_1 = 2n_3 + n_2 + n_4$. Rearranging the columns of $Z$, we obtain
\begin{equation}
\label{eq:all}
\widehat{Z} = \begin{bmatrix}
p_ 1& Rp_1 & p_2 & Rp_2 & \cdots & p_l & Rp_l
\end{bmatrix}
\end{equation}
Then, using Lemmas \ref{thm:conjugate_eig}, \ref{thm:pairs}, \ref{thm:repeated}, the $2l^{\text{th}}$ degree term in the determinant of $\widehat{Z}$ of the form
\begin{align*} \prod_{i = 1}^{n_4} P_{n_3+i,1,n_3+i}^2  \prod_{j = 1}^{n_2} P_{n_1+j,1,n_1+j}P_{n_1+j,2,n_1+j}  \prod_{k_1 = 1}^{n_3}\prod_{k_2 = 1,2} P_{n_1 + n_2 + k_1,k_2, (k_1 - 1)*2+k_2} P_{k_1,2,(k_1-1)*2+k_2} \\
=    \prod_{i = 1}^{n_4}  \prod_{j = 1}^{n_2}\prod_{k_1 = 1}^{n_3} \prod_{k_2 = 1,2} P_{n_3+i,1,n_3+i}^2 P_{n_1+j,1,n_1+j} P_{n_1+j,2,n_1+j} P_{n_1 + n_2 + k_1,k_2, (k_1 - 1)*2+k_2} P_{k_1,2,(k_1-1)*2+k_2} 
\end{align*}
 has a coefficient given by 
\[ 
\prod_{i = 1}^{n_4} \sin\theta_{n_3+i}\prod_{j = 1}^{n_2} \lambda_{n_1+j,2} - \lambda_{n_1+j,1} \prod_{k = 1}^{n_3} \sin^2\theta_{k}  =  \prod_{i = 1}^{n_4} \prod_{j = 1}^{n_2}\prod_{k = 1}^{n_3} \sin\theta_{n_3+i} \sin^2\theta_{k} (\lambda_{n_1+j,2} - \lambda_{n_1+j,1}) \neq 0.  \]
Therefore, there is a $2l \times 2l$ submatrix in $Z$ with non-zero determinant, which renders $Z$ full column rank.  We have hence proved Theorem \ref{thm:final_k2} for the case where $A_i$'s are orthogonal and span the same subspace. 

The same result can be extended to any general square matrix $U = A^TB$. The proof for this can be obtained using methods exactly as above, but by operating on the real Schur decomposition of $U$. The real Schur decomposition of $U$ can be expressed as 
\[ U = QRQ^T\] where $R$ is no longer block diagonal but is a pseudo-upper triangular matrix with either $2 \times 2$ or $1 \times 1$ blocks along the diagonal that reflect the eigenvalue structure of $U$.

\section{Proof of Theorem \ref{thm:manyK}}

Again for ease of notation, we prove the result on the transposed matrices. We now provide a sufficient condition on an ensemble of matrices $A_1, A_2, \cdots A_k \in \R^{d \times n}$ such that for a random matrix $P$ of size $n \times l$, the matrix $M = \begin{bmatrix}
A_1P & A_2P & \cdots & A_kP
\end{bmatrix} $ is full column rank. 

Towards this end, note that we have
\[
M = \begin{bmatrix}
A_1P & A_2P & \cdots & A_kP
\end{bmatrix} \\
= \begin{bmatrix}
A_1VV^TP & A_2VV^TP & \cdots & A_kVV^TP
\end{bmatrix}
\]
Due to the rotational invariance of the standard normal distribution, we can denote $V^TP$ as $P$ and also denote $A_iV$ as $A_i$ itself without loss of generality. With this in place, we can rearrange the columns of $M$ as 
\[ \widehat{M} = \begin{bmatrix}
A_1p_1 & A_2p_1 & \cdots & A_kp_1 & A_1p_2 & \cdots & A_kp_l
\end{bmatrix} \]
Expanding the determinant of the first $kl \times kl$ submatrix of $\widehat{M}$, we can obtain the coefficient of the term $p_{11}^k p_{22}^k\cdots p_{ll}^k$ as 
\begin{align*}
C & = \det(\begin{bmatrix}
A_1^{(1)} & A_2^{(1)} \cdots A_k^{(1)} & A_1^{(2)} & \cdots & A_k^{(l)}
\end{bmatrix} \\
& = \det(\begin{bmatrix}
A_{1\mathcal{S}} & A_{2\mathcal{S}} & \cdots & A_{k\mathcal{S}})
\end{bmatrix} \\
& = \det(\begin{bmatrix}
A _1V_{\mathcal{S}} & A_2V_{\mathcal{S}} & \cdots & A_kV_{\mathcal{S}}
\end{bmatrix})
\end{align*} where we assume $\mathcal{S} = \{ 1,\cdots, l\}$ without loss of generality and $A_i^{(j)}$ denotes the $j^{\text{th}}$ column of $A_i$. If $\begin{bmatrix}
A _1V_{\mathcal{S}} & A_2V_{\mathcal{S}} & \cdots & A_kV_{\mathcal{S}}
\end{bmatrix}$  is full column rank, then $C \neq 0$ and the result follows. 

When $r$ is not a multiple of $K$, the same result can be extended to show that if there exist index sets $\mathcal{S}_1$ and $\mathcal{S}_2$ such that $|\mathcal{S}_2| = \lfloor r/K \rfloor$, $|\mathcal{S}_1|= \lfloor r/K\rfloor+1$, $\mathcal{S}_2 \subset \mathcal{S}_1$ and an orthobasis $V$ such that $ \widehat{M} = \begin{bmatrix}
A _1V_{\mathcal{S}_1} & A_2V_{\mathcal{S}_1} & \cdots  & A_qV_{\mathcal{S}_1} & A_{q+1}V_{\mathcal{S}_2}& A_kV_{\mathcal{S}_2}
\end{bmatrix}$ is full column rank, then for $l = \lfloor r/K \rfloor + 1$, $\calR(Z) = \calR(A)$. 

\clearpage
\bibliographystyle{unsrt}
\bibliography{bibliography}

\end{document}